\documentclass[conference,letterpaper]{IEEEtran}

\usepackage[margin=0.55in]{geometry}
\usepackage[utf8]{inputenc} 
\usepackage[T1]{fontenc}

\ifCLASSINFOpdf

\else

\fi
\usepackage[cmex10]{amsmath}
\usepackage{ifthen}
\usepackage{cite}
\usepackage{amsthm,amsfonts,amssymb}
\usepackage{graphicx}
\usepackage{subfig}
\usepackage{blkarray}
\usepackage[mathscr]{euscript}
\usepackage{enumitem}
\usepackage{algorithm}
\usepackage[noend]{algpseudocode}
\usepackage{blkarray}
\usepackage{stfloats}%
\usepackage{xfrac}
\usepackage{faktor}
\newtheorem{theorem}{Theorem}[section]

\newtheorem{proposition}{Proposition}[section]
\newtheorem{lemma}{Lemma}[section]
\theoremstyle{remark}
\newtheorem*{remark}{Remark}
\theoremstyle{definition}
\newtheorem{definition}{Definition}[section]
\usepackage{url}
\usepackage{array}
\usepackage{arydshln}
\usepackage{mathtools}
\usepackage{dsfont}
\usepackage{physics}
\usepackage{tikz}
\usepackage{mathdots}
\usepackage{cancel}
\usepackage{color}
\usepackage{siunitx}

\usepackage{multirow}
\usepackage{amssymb}
\usepackage{gensymb}
\usepackage{tabularx}
\usepackage{extarrows}
\usepackage{booktabs}
\usetikzlibrary{fadings}
\usetikzlibrary{patterns}
\usetikzlibrary{shadows.blur}
\usetikzlibrary{shapes}

\usepackage[cmintegrals]{newtxmath}

\begin{document}
%
\title{Linear Runlength-Limited Subcodes of Reed-Muller Codes and Coding Schemes for Input-Constrained BMS Channels}
%
%
%
\author{V.~Arvind Rameshwar \ and \ Navin Kashyap
	\thanks{The authors are with the Department of Electrical Communication Engineering, Indian Institute of Science, Bengaluru 560012. Email: \{\texttt{vrameshwar}, \texttt{nkashyap}\}\texttt{@iisc.ac.in}}
	\thanks{The work of V.~A.~Rameshwar was supported by a Prime Minister's Research Fellowship, from the Ministry of Education, Govt. of India.}
}
\IEEEoverridecommandlockouts

\maketitle

\begin{abstract}
In this work, we address the question of the largest rate of linear subcodes of Reed-Muller (RM) codes, all of whose codewords respect a runlength-limited (RLL) constraint. Our interest is in the $(d,\infty)$-RLL constraint, which mandates that every pair of successive $1$s be separated by at least $d$ $0$s. Consider any sequence $\{{\mathcal{C}_m}\}_{m\geq 1}$ of RM codes with increasing blocklength, whose rates approach $R$, in the limit as the blocklength goes to infinity. We show that for any linear $(d,\infty)$-RLL subcode, $\hat{\mathcal{C}}_m$, of the code $\mathcal{C}_m$, it holds that the rate of $\hat{\mathcal{C}}_m$ is at most $\frac{R}{d+1}$, in the limit as the blocklength goes to infinity. We also consider scenarios where the coordinates of the RM codes are not ordered according to the standard lexicographic ordering, and derive rate upper bounds for linear $(d,\infty)$-RLL subcodes, in those cases as well. Next, for the setting of a $(d,\infty)$-RLL input-constrained binary memoryless symmetric (BMS) channel, we devise a new coding scheme, based on cosets of RM codes. Again, in the limit of blocklength going to infinity, this code outperforms any linear subcode of an RM code, in terms of rate, for low noise regimes of the channel.
\end{abstract}


%
\IEEEpeerreviewmaketitle

\section{Introduction}
\label{sec:intro}
%
%
%
%
The physical limitations of hardware used in most data recording and communication systems cause some sequences to be more prone to error than others.  Constrained coding is a method of alleviating this problem, by encoding arbitrary user data sequences into sequences that respect a constraint (see, for example, \cite{Roth} or \cite{Immink}). In this work, we investigate the sizes of linear subcodes of well-known families of codes, all of whose codewords obey a certain hard constraint. In particular, we work with the binary Reed-Muller (RM) family of codes and obtain upper bounds on the sizes of linear subcodes that obey a runlength-limited (RLL) constraint.

The specific hard constraint of interest to us is the $(d,\infty)$-RLL constraint, which admits only binary sequences with at least $d$ $0$s between every pair of successive $1$s. Figure 2 shows a state transition graph that represents the constraint. This constraint is a special case of the $(d,k)$-RLL constraint, which admits only binary sequences with at least $d$ and at most $k$ $0$s between successive $1$s. 
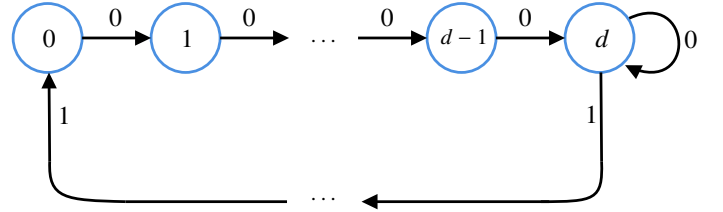
\begin{figure}[!h]

\begin{center}
\resizebox{0.5\textwidth}{!}{
\tikzset{every picture/.style={line width=0.75pt}} 

\begin{tikzpicture}[x=0.75pt,y=0.75pt,yscale=-1,xscale=1]


\draw  [color={rgb, 255:red, 74; green, 144; blue, 226 }  ,draw opacity=1 ][line width=1.5]  (130,96) .. controls (130,82.19) and (141.19,71) .. (155,71) .. controls (168.81,71) and (180,82.19) .. (180,96) .. controls (180,109.81) and (168.81,121) .. (155,121) .. controls (141.19,121) and (130,109.81) .. (130,96) -- cycle ;
\draw  [color={rgb, 255:red, 74; green, 144; blue, 226 }  ,draw opacity=1 ][line width=1.5]  (230,96) .. controls (230,82.19) and (241.19,71) .. (255,71) .. controls (268.81,71) and (280,82.19) .. (280,96) .. controls (280,109.81) and (268.81,121) .. (255,121) .. controls (241.19,121) and (230,109.81) .. (230,96) -- cycle ;
\draw  [color={rgb, 255:red, 74; green, 144; blue, 226 }  ,draw opacity=1 ][line width=1.5]  (430,96) .. controls (430,82.19) and (441.19,71) .. (455,71) .. controls (468.81,71) and (480,82.19) .. (480,96) .. controls (480,109.81) and (468.81,121) .. (455,121) .. controls (441.19,121) and (430,109.81) .. (430,96) -- cycle ;
\draw  [color={rgb, 255:red, 74; green, 144; blue, 226 }  ,draw opacity=1 ][line width=1.5]  (530,96) .. controls (530,82.19) and (541.19,71) .. (555,71) .. controls (568.81,71) and (580,82.19) .. (580,96) .. controls (580,109.81) and (568.81,121) .. (555,121) .. controls (541.19,121) and (530,109.81) .. (530,96) -- cycle ;
\draw [line width=1.5]    (180,96) -- (226,96) ;
\draw [shift={(230,96)}, rotate = 180] [fill={rgb, 255:red, 0; green, 0; blue, 0 }  ][line width=0.08]  [draw opacity=0] (11.61,-5.58) -- (0,0) -- (11.61,5.58) -- cycle    ;
\draw [line width=1.5]    (280,96) -- (326,96) ;
\draw [shift={(330,96)}, rotate = 180] [fill={rgb, 255:red, 0; green, 0; blue, 0 }  ][line width=0.08]  [draw opacity=0] (11.61,-5.58) -- (0,0) -- (11.61,5.58) -- cycle    ;
\draw [line width=1.5]    (380,96) -- (426,96) ;
\draw [shift={(430,96)}, rotate = 180] [fill={rgb, 255:red, 0; green, 0; blue, 0 }  ][line width=0.08]  [draw opacity=0] (11.61,-5.58) -- (0,0) -- (11.61,5.58) -- cycle    ;
\draw [line width=1.5]    (480,96) -- (526,96) ;
\draw [shift={(530,96)}, rotate = 180] [fill={rgb, 255:red, 0; green, 0; blue, 0 }  ][line width=0.08]  [draw opacity=0] (11.61,-5.58) -- (0,0) -- (11.61,5.58) -- cycle    ;
\draw [line width=1.5]    (556,121) -- (556.5,185.73) ;
\draw [line width=1.5]    (156.03,125) -- (156.5,185.73) ;
\draw [shift={(156,121)}, rotate = 89.56] [fill={rgb, 255:red, 0; green, 0; blue, 0 }  ][line width=0.08]  [draw opacity=0] (11.61,-5.58) -- (0,0) -- (11.61,5.58) -- cycle    ;
\draw [line width=1.5]    (156.5,185.73) .. controls (155.5,212.73) and (167.5,213.73) .. (211.5,214.73) ;
\draw [line width=1.5]    (556.5,185.73) .. controls (557.5,213.73) and (551.5,214.73) .. (512.5,214.73) ;
\draw [line width=1.5]    (211.5,214.73) -- (328.5,214.73) ;
\draw [line width=1.5]    (386.5,214.73) -- (518.5,214.73) ;
\draw [shift={(382.5,214.73)}, rotate = 0] [fill={rgb, 255:red, 0; green, 0; blue, 0 }  ][line width=0.08]  [draw opacity=0] (11.61,-5.58) -- (0,0) -- (11.61,5.58) -- cycle    ;
\draw [line width=1.5]    (576.5,81.73) .. controls (622.56,60.17) and (626.36,138.5) .. (576.61,117.17) ;
\draw [shift={(573.5,115.73)}, rotate = 386.13] [fill={rgb, 255:red, 0; green, 0; blue, 0 }  ][line width=0.08]  [draw opacity=0] (11.61,-5.58) -- (0,0) -- (11.61,5.58) -- cycle    ;

\draw (344,96) node [anchor=north west][inner sep=0.75pt]  [font=\large]  {$\dotsc $};
\draw (344,210.7) node [anchor=north west][inner sep=0.75pt]  [font=\large]  {$\dotsc $};
\draw (150,88.4) node [anchor=north west][inner sep=0.75pt] [font=\Large]   {$0$};
\draw (250,88.4) node [anchor=north west][inner sep=0.75pt]   [font=\Large] {$1$};
\draw (438,88.4) node [anchor=north west][inner sep=0.75pt]   [font=\large] {$d-1$};
\draw (550,88.4) node [anchor=north west][inner sep=0.75pt]  [font=\Large]  {$d$};
\draw (198,73.4) node [anchor=north west][inner sep=0.75pt]  [font=\Large]  {$0$};
\draw (295,73.4) node [anchor=north west][inner sep=0.75pt]  [font=\Large]  {$0$};
\draw (395,73.4) node [anchor=north west][inner sep=0.75pt]  [font=\Large]  {$0$};
\draw (495,73.4) node [anchor=north west][inner sep=0.75pt]  [font=\Large]  {$0$};
\draw (615,89.4) node [anchor=north west][inner sep=0.75pt] [font=\Large]   {$0$};
\draw (543,143.4) node [anchor=north west][inner sep=0.75pt]  [font=\Large]  {$1$};
\draw (161,144.4) node [anchor=north west][inner sep=0.75pt] [font=\Large]   {$1$};

\end{tikzpicture}
}
\end{center}
\label{fig:d_inf}
\caption{The state transition graph for the $(d,\infty)$-RLL constraint.}
\end{figure}

One of the motivations for studying this problem is the design of explicit coding schemes that achieve good rates over input-constrained discrete memoryless channels (DMCs). 
Figure \ref{fig:gen_const_DMC} shows a generic binary memoryless symmetric (BMS) channel with input constraints. Input-constrained DMCs in general fall under the broad class of discrete finite-state channels (DFSCs, or FSCs).

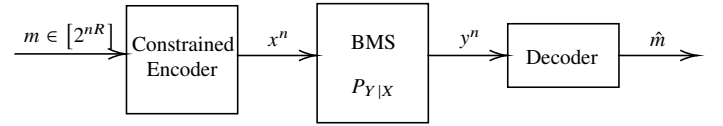
\begin{figure}[!t]
	\centering
	\resizebox{0.5\textwidth}{!}{

		\tikzset{every picture/.style={line width=0.75pt}} 
		
		\begin{tikzpicture}[x=0.75pt,y=0.75pt,yscale=-1,xscale=1]
			
			\draw   (321,106.32) -- (391,106.32) -- (391,174.32) -- (321,174.32) -- cycle ;
			\draw    (390.5,137.65) -- (438.5,137.65) ;
			\draw [shift={(440.5,137.65)}, rotate = 180] [color={rgb, 255:red, 0; green, 0; blue, 0 }  ][line width=0.75]    (10.93,-3.29) .. controls (6.95,-1.4) and (3.31,-0.3) .. (0,0) .. controls (3.31,0.3) and (6.95,1.4) .. (10.93,3.29)   ;
			\draw   (441,105) -- (511,105) -- (511,180) -- (441,180) -- cycle ;
			\draw   (561,117.65) -- (631,117.65) -- (631,157.65) -- (561,157.65) -- cycle ;
			\draw    (510.5,137.65) -- (558.5,137.65) ;
			\draw [shift={(560.5,137.65)}, rotate = 180] [color={rgb, 255:red, 0; green, 0; blue, 0 }  ][line width=0.75]    (10.93,-3.29) .. controls (6.95,-1.4) and (3.31,-0.3) .. (0,0) .. controls (3.31,0.3) and (6.95,1.4) .. (10.93,3.29)   ;
			\draw    (630.5,137.65) -- (678.5,137.65) ;
			\draw [shift={(680.5,137.65)}, rotate = 180] [color={rgb, 255:red, 0; green, 0; blue, 0 }  ][line width=0.75]    (10.93,-3.29) .. controls (6.95,-1.4) and (3.31,-0.3) .. (0,0) .. controls (3.31,0.3) and (6.95,1.4) .. (10.93,3.29)   ;
			\draw    (250.5,136.65) -- (318.5,136.65) ;
			\draw [shift={(320.5,136.65)}, rotate = 180] [color={rgb, 255:red, 0; green, 0; blue, 0 }  ][line width=0.75]    (10.93,-3.29) .. controls (6.95,-1.4) and (3.31,-0.3) .. (0,0) .. controls (3.31,0.3) and (6.95,1.4) .. (10.93,3.29)   ;

			\draw (461,151.01) node [anchor=north west][inner sep=0.75pt]  [font=\normalsize]  {$P_{Y|X}$};
			\draw (649,121.01) node [anchor=north west][inner sep=0.75pt]  [font=\normalsize]  {$\hat{m}$};
			\draw (529,120.01) node [anchor=north west][inner sep=0.75pt]  [font=\normalsize]  {$y^{n}$};
			\draw (409,121.01) node [anchor=north west][inner sep=0.75pt]  [font=\normalsize]  {$x^{n}$};
			\draw (571,132.61) node [anchor=north west][inner sep=0.75pt]  [font=\normalsize] [align=left] {Decoder};
			\draw (323,124) node [anchor=north west][inner sep=0.75pt]   [align=left] {{\normalsize Constrained}\\{\normalsize \ \ Encoder}};
			\draw (461,123) node [anchor=north west][inner sep=0.75pt]   [align=left] {BMS};
			\draw (255,116) node [anchor=north west][inner sep=0.75pt]  [font=\normalsize]  {${\displaystyle m\in \left[ 2^{nR}\right]}$};

		\end{tikzpicture}
}	
	\caption{System model of an input-constrained binary memoryless symmetric (BMS) channel without feedback.}
	\label{fig:gen_const_DMC}
\end{figure}

While explicit codes achieving the capacities or whose rates are very close to the capacities of unconstrained DMCs have been derived in works such as \cite{polar, kud1, luby, ru1, kud2}, the problem of designing coding schemes for input-constrained DMCs has not received much attention in the literature. Moreover, unlike the case of the unconstrained DMC, whose capacity is characterized by Shannon's  single-letter, computable formula, $C_{\text{DMC}} = \sup_{P(x)} I(X;Y)$, the explicit computation of the capacity of an FSC is a much more difficult problem to tackle. 

With the recent result of Reeves and Pfister \cite{Reeves} that Reed-Muller (RM) codes achieve the capacity of the unconstrained BMS channel under bit-MAP decoding, there opens the possibility of using such algebraic codes over input-constrained BMS channels as well. Suppose that $C$ is the capacity of the unconstrained channel.
The authors in \cite{arnk22arxiv} showed a simple linear coding scheme, using subcodes of RM codes, with rates of $C\cdot{2^{-\left \lceil \log_2(d+1)\right \rceil}}$, being achievable over $(d,\infty)$-RLL input-constrained BMS channels. In this paper, we prove that 
any linear RM subcode that respects the $(d,\infty)$-RLL constraint, must have a rate of at most $\frac{C}{d+1}$, in the limit as the blocklength goes to infinity. In doing so, we show that one cannot do better, asymptotically, than the simple coding scheme in \cite{arnk22arxiv}, if one requires that the subcodes be linear. We also consider the rates achieved using linear $(d,\infty)$-RLL subcodes of permuted RM codes, and show that for codes of large enough blocklength, almost all permutations must respect an upper bound of $\frac{C}{d+1}+\delta$, for $\delta$ being as small as is required.

As an improvement over the rates achievable using linear $(d,\infty)$-RLL subcodes of RM codes, we propose a new coding scheme that uses cosets of RM codes. The rate achieved by this scheme is $\frac{C_0\cdot C^2\cdot 2^{-\left \lceil \log_2(d+1)\right \rceil}}{C^2\cdot 2^{-\left \lceil \log_2(d+1)\right \rceil} + 1-C+\epsilon}$, where $C_0$ is the noiseless capacity of the input constraint, and $\epsilon>0$ can be taken to be as small as is required. For example, when $d=1$, the rates achieved using this cosets-based scheme are better than those achieved by any scheme that uses linear $(1,\infty)$-RLL subcodes of RM codes, when $C\gtrapprox 0.7613$. Moreover, as the capacity of the channel approaches $1$, i.e., as the channel noise approaches $0$, the rate achieved by our cosets-based scheme approaches a value arbitrarily close to $C_0$, which is the largest rate achievable, at zero noise, given the constraint.


Our results supplement the analysis in \cite{pvk}, on rates achievable by $(d,k)$-RLL subcodes of cosets of a linear block code. Specifically, Corollary 1 of \cite{pvk} provides an existence result on cosets of capacity-achieving (over the unconstrained BMS channel) codes, whose constrained subcodes have rate at least $C_0 + C -1$. The coding scheme in this paper achieves rates close to the lower bound in \cite{pvk}, for values of $C$ close to $1$. We note that using linear $(d,\infty)$-RLL subcodes of RM codes as in \cite{arnk22arxiv}, we can achieve larger rates as compared to the rate lower bound in \cite{pvk}, when the capacity $C$ is low, i.e., when $C<(1-C_0)\cdot \left(1-{2^{-\left \lceil \log_2(d+1)\right \rceil}}\right)^{-1}$. 


The remainder of the paper is organized as follows: Section \ref{sec:notation} introduces the notation and provides the necessary background. Section \ref{sec:main} states our main results. Section \ref{sec:rmublin} discusses upper bounds on the rate achievable over the BMS channel, using linear $(d,\infty)$-RLL subcodes. In Section \ref{sec:perm}, the question of upper bounds on rates achievable using linear $(d,\infty)$-RLL subcodes, under coordinate orderings different from the standard lexicographic ordering, is taken up. Section \ref{sec:cosets} then discusses a construction that uses cosets of RM codes to achieve good rates.
Finally, Section \ref{sec:conclusion} contains concluding remarks and a discussion on possible future work.

\section{Notation and Preliminaries}
\label{sec:notation}
\subsection{Notation}

Random variables will be denoted by capital letters, and their realizations by lower-case letters, e.g., $X$ and $x$, respectively. Calligraphic letters, e.g., $\mathscr{X}$, denote sets. The notation $[n]$ denotes the set, $\{1,2,\ldots,n\}$, of integers, and the notation $[a:b]$, for $a<b$, denotes the set of integers $\{a,a+1,\ldots,b\}$. Moreover, for a real number $x$, we use $\left \lfloor x \right \rfloor$ to denote the largest integer smaller than or equal to $x$. For vectors $\mathbf{w}$ and $\mathbf{v}$ of length $n$ and $m$, respectively, we denote their concatenation by the $(m+n)$-length vector, $\mathbf{w}\mathbf{v}$. The notation 
$x^N$ denotes the vector $(x_1,\ldots,x_N)$. We also use the notation $\mathbf{e}_i^{(n)}$ to denote the standard basis vector of length $n$, with a $1$ at position $i$, and $0$s elsewhere, for $i\in [n]$. Further, we denote by $S_{(d,\infty)}^{(n)}$, the set of all $n$-length binary words that respect the $(d,\infty)$-RLL constraint, and we set $S_{(d,\infty)}=\bigcup_{n\geq 1} S_{(d,\infty)}^{(n)}$. 

All logarithms are to the base $2$. Throughout, we use the convenient notation $\binom{m}{\le r}$ to denote the summation $\sum\limits_{i=0}^r \binom{m}{i}$, and the notation $\binom{m}{\ge r}$ to denote $\sum\limits_{i=r}^m \binom{m}{i}$.

\subsection{Reed-Muller Codes}
\label{sec:introrm}
We recall the definition of the binary Reed-Muller (RM) family of codes. Codewords of binary RM codes consist of the evaluation vectors of multivariate polynomials over the binary field $\mathbb{F}_2$. Consider the polynomial ring $\mathbb{F}_2[x_1,x_2,\ldots,x_m]$ in $m$ variables. Note that in the specification of a polynomial $f\in \mathbb{F}_2[x_1,x_2,\ldots,x_m]$, only monomials of the form $\prod_{j\in S} x_j$, for some $S\subseteq [m]$, need to be considered, since $x^2 = x$ over the field $\mathbb{F}_2$, for an indeterminate $x$. For a polynomial $f\in \mathbb{F}_2[x_1,x_2,\ldots,x_m]$ and a binary vector $\mathbf{z} = (z_1,\ldots,z_m)\in \mathbb{F}_2^m$, let Eval$_\mathbf{z}(f):=f(z_1,\ldots,z_m)$. We let the evaluation points be ordered according to the standard lexicographic order on strings in $\mathbb{F}_2^m$, i.e., if $\mathbf{z} = (z_1,\ldots,z_m)$ and $\mathbf{z}^{\prime} = (z_1^{\prime},\ldots,z_m^{\prime})$ are two distinct evaluation points, then, $\mathbf{z}$ occurs before $\mathbf{z}^{\prime}$ in our ordering if and only if, for some $i\geq 1$, it holds that $z_j = z_j^{\prime}$ for all $j<i$, and $z_i < z_i^{\prime}$. Now, let Eval$(f):=\left(\text{Eval}_\mathbf{z}(f):\mathbf{z}\in \mathbb{F}_2^m\right)$ be the evaluation vector of $f$, where the coordinates $\mathbf{z}$ are ordered according to the standard lexicographic order. 

\begin{definition}[see \cite{mws}, Chap. 13, or \cite{rm_survey}]
	The $r^{\text{th}}$ order binary Reed-Muller code RM$(m,r)$ is defined as the set of binary vectors:
	\[
	\text{RM}(m,r):=\bigl\{\text{Eval}(f): f\in \mathbb{F}_2[x_1,x_2,\ldots,x_m],\ \text{deg}(f)\leq r\bigr\},
	\]
	where $\text{deg}(f)$ is the degree of the largest monomial in $f$, and the degree of a monomial $\prod_{j\in S} x_j$ is simply $|S|$. 
\end{definition}

It is well-known that RM$(m,r)$ has dimension $\binom{m}{\le r}$ and minimum Hamming distance $2^{m-r}$. The weight of a codeword $\mathbf{c} = \text{Eval}(f)$ is the number of $1$s in its evaluation vector, i.e,
\[
\text{wt}\left(\text{Eval}(f)\right):=|\{\mathbf{z}\in \mathbb{F}_2^m: f(\mathbf{z})=1\}|.
\]
In what follows, we let $G_{\text{Lex}}(m,r)$ be the generator matrix of $\text{RM}(m,r)$ consisting of rows that are the evaluations, in the lexicographic order, of monomials of degree less than or equal to $r$. The columns of $G_{\text{Lex}}(m,r)$ will be indexed by $m$-tuples $\mathbf{b} = (b_1,\ldots,b_m)$ in the lexicographic order.


\subsection{Codes for BMS Channels}

The communication setting of an input-constrained binary memoryless symmetric (BMS) channel without feedback is shown in Figure 2. A message $M$ is drawn uniformly from the set $\{1,2,\ldots,2^{nR}\}$, and is made available to the constrained encoder. The encoder produces a binary input sequence $x^n \in \{0,1\}^n = \mathscr{X}^n$, which is constrained to obey the $(d,\infty)$-RLL input constraint, a state transition graph for which is shown in Figure 2. Note that $d=0$ corresponds to the absence of any constraint.

The channel output alphabet is the extended real line, i.e., $\mathscr{Y} = \overline{\mathbb{R}}$. The channel is memoryless in the sense that $P(y_i|x^{i},y^{i-1}) = P(y_i|x_i)$, for all $i$. Further, the channel is symmetric, in that $P(y|1) = P(-y|0)$, for all $y\in \mathscr{Y}$. 
Common examples of BMS channels include the binary erasure channel (BEC$(\epsilon)$), the binary symmetric channel (BSC), and the binary additive white Gaussian noise (BI-AWGN) channel. Figures \ref{fig:bec} and \ref{fig:bsc} depict the BEC and BSC, pictorially. 

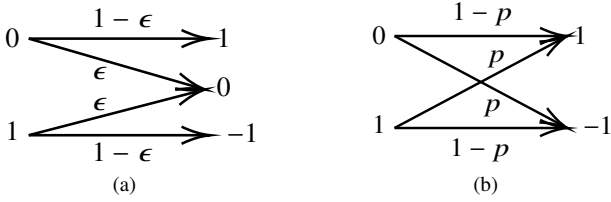
\begin{figure}[!h]
 \centering
\subfloat[]{
\resizebox{0.21\textwidth}{!}{

\tikzset{every picture/.style={line width=0.75pt}} 

\begin{tikzpicture}[x=0.75pt,y=0.75pt,yscale=-1,xscale=1]
	
	\draw    (454.24,136.67) -- (521.85,136.67) ;
	\draw [shift={(523.85,136.67)}, rotate = 180] [color={rgb, 255:red, 0; green, 0; blue, 0 }  ][line width=0.75]    (10.93,-3.29) .. controls (6.95,-1.4) and (3.31,-0.3) .. (0,0) .. controls (3.31,0.3) and (6.95,1.4) .. (10.93,3.29)   ;
	\draw    (454.24,136.67) -- (519.94,155.67) ;
	\draw [shift={(521.86,156.22)}, rotate = 196.13] [color={rgb, 255:red, 0; green, 0; blue, 0 }  ][line width=0.75]    (10.93,-3.29) .. controls (6.95,-1.4) and (3.31,-0.3) .. (0,0) .. controls (3.31,0.3) and (6.95,1.4) .. (10.93,3.29)   ;
	\draw    (454.24,173.79) -- (521.85,173.79) ;
	\draw [shift={(523.85,173.79)}, rotate = 180] [color={rgb, 255:red, 0; green, 0; blue, 0 }  ][line width=0.75]    (10.93,-3.29) .. controls (6.95,-1.4) and (3.31,-0.3) .. (0,0) .. controls (3.31,0.3) and (6.95,1.4) .. (10.93,3.29)   ;
	\draw    (454.24,173.79) -- (519.93,156.72) ;
	\draw [shift={(521.86,156.22)}, rotate = 165.43] [color={rgb, 255:red, 0; green, 0; blue, 0 }  ][line width=0.75]    (10.93,-3.29) .. controls (6.95,-1.4) and (3.31,-0.3) .. (0,0) .. controls (3.31,0.3) and (6.95,1.4) .. (10.93,3.29)   ;
	
	\draw (476.48,146) node [anchor=north west][inner sep=0.75pt]  [font=\footnotesize]  {$\epsilon $};
	\draw (476.48,158) node [anchor=north west][inner sep=0.75pt]  [font=\footnotesize]  {$\epsilon $};
	\draw (477.13,123.4) node [anchor=north west][inner sep=0.75pt]  [font=\footnotesize]  {$1-\epsilon $};
	\draw (477.13,175) node [anchor=north west][inner sep=0.75pt]  [font=\footnotesize]  {$1-\epsilon $};
	\draw (443.86,131.32) node [anchor=north west][inner sep=0.75pt]  [font=\footnotesize]  {${\displaystyle 0}$};
	\draw (443.86,166.96) node [anchor=north west][inner sep=0.75pt]  [font=\footnotesize]  {$1$};
	\draw (525.01,131.32) node [anchor=north west][inner sep=0.75pt]  [font=\footnotesize]  {$1$};
	\draw (525.01,149.88) node [anchor=north west][inner sep=0.75pt]  [font=\footnotesize]  {$0$};
	\draw (525.01,167.7) node [anchor=north west][inner sep=0.75pt]  [font=\footnotesize]  {$\ -1$};

\end{tikzpicture}
}
\label{fig:bec}
}
\qquad
\subfloat[]{
\resizebox{0.2\textwidth}{!}{

\tikzset{every picture/.style={line width=0.75pt}} 

\begin{tikzpicture}[x=0.75pt,y=0.75pt,yscale=-1,xscale=1]
	
	\draw    (449.24,119.67) -- (516.85,119.67) ;
	\draw [shift={(518.85,119.67)}, rotate = 180] [color={rgb, 255:red, 0; green, 0; blue, 0 }  ][line width=0.75]    (10.93,-3.29) .. controls (6.95,-1.4) and (3.31,-0.3) .. (0,0) .. controls (3.31,0.3) and (6.95,1.4) .. (10.93,3.29)   ;
	\draw    (449.24,119.67) -- (517.09,155.85) ;
	\draw [shift={(518.85,156.79)}, rotate = 208.07] [color={rgb, 255:red, 0; green, 0; blue, 0 }  ][line width=0.75]    (10.93,-3.29) .. controls (6.95,-1.4) and (3.31,-0.3) .. (0,0) .. controls (3.31,0.3) and (6.95,1.4) .. (10.93,3.29)   ;
	\draw    (449.24,156.79) -- (516.85,156.79) ;
	\draw [shift={(518.85,156.79)}, rotate = 180] [color={rgb, 255:red, 0; green, 0; blue, 0 }  ][line width=0.75]    (10.93,-3.29) .. controls (6.95,-1.4) and (3.31,-0.3) .. (0,0) .. controls (3.31,0.3) and (6.95,1.4) .. (10.93,3.29)   ;
	\draw    (449.24,156.79) -- (517.09,120.61) ;
	\draw [shift={(518.85,119.67)}, rotate = 151.93] [color={rgb, 255:red, 0; green, 0; blue, 0 }  ][line width=0.75]    (10.93,-3.29) .. controls (6.95,-1.4) and (3.31,-0.3) .. (0,0) .. controls (3.31,0.3) and (6.95,1.4) .. (10.93,3.29)   ;
	
	\draw (485.05,124) node [anchor=north west][inner sep=0.75pt]  [font=\footnotesize]  {$p$};
	\draw (483.48,144) node [anchor=north west][inner sep=0.75pt]  [font=\footnotesize]  {$p$};
	\draw (471.13,104.4) node [anchor=north west][inner sep=0.75pt]  [font=\footnotesize]  {$1-p$};
	\draw (470.13,159.4) node [anchor=north west][inner sep=0.75pt]  [font=\footnotesize]  {$1-p$};
	\draw (438.86,114.32) node [anchor=north west][inner sep=0.75pt]  [font=\footnotesize]  {${\displaystyle 0}$};
	\draw (438.86,149.96) node [anchor=north west][inner sep=0.75pt]  [font=\footnotesize]  {$1$};
	\draw (520.01,114.32) node [anchor=north west][inner sep=0.75pt]  [font=\footnotesize]  {$1$};
	\draw (520.01,150.7) node [anchor=north west][inner sep=0.75pt]  [font=\footnotesize]  {$\ -1$};

\end{tikzpicture}
}
\label{fig:bsc}}
\caption{(a) The binary erasure channel (BEC$(\epsilon)$) with erasure probability $\epsilon$ and output alphabet $\mathscr{Y} = \{-1,0,1\}$, with the output symbol $0$ denoting an erasure. (b) The binary symmetric channel (BSC$(p)$) with crossover probability $p$ and output alphabet $\mathscr{Y} = \{-1,1\}$.}
\end{figure}

\begin{definition}
	\label{def:ach}
	An $(n,2^{nR},(d,\infty))$ code for an input-constrained channel {without feedback} is defined by the encoding function:
	\begin{equation}
		\label{eq:encoder}
		f: \{1,\ldots, 2^{nR}\}\rightarrow \mathscr{X}^n, \quad i\in [n],
	\end{equation}
	such that $(x_{i+1},\ldots,x_{\min\{i+d,n\}}) = (0,\ldots,0)$, if $x_i = 1$. 
	
	Given an output sequence $y^n$, the bit-MAP decoder $\Psi: \mathscr{Y}^n\rightarrow \mathscr{X}^n$ outputs $\hat{\mathbf{x}}:=(\hat{x}_1,\ldots,\hat{x}_n)$, where, for each $i\in [n]$, the estimate
	\begin{equation*}
		\hat{x}_i:=\text{argmax}_{x\in \{0,1\}}  P(X_i=x|y^n).
	\end{equation*}
Likewise, the block-MAP decoder $\Phi: \mathscr{Y}^n\rightarrow \mathscr{X}^n$ outputs as estimate
	\begin{equation*}
		\hat{x}^n:=\text{argmax}_{x^n\in \{0,1\}^n}  P(X^n=x^n|y^n).
	\end{equation*}
	The error under bit-MAP decoding is defined as $$P_b^{(n)}:=1-\frac{1}{n} \sum_{i=1}^{n}\mathbb{E}[\max\{P(X_i=0|Y^n),P(X_i=1|Y^n)\}],$$
	 and the error under block-MAP decoding is defined as $$P_B^{(n)}:=P(\Phi(Y^n)\neq X^n).$$
	A rate $R$ is said to be $(d,\infty)$-achievable under bit-MAP decoding, if there exists a sequence of $(n,2^{nR_n},(d,\infty))$ codes, $\{\mathcal{C}^{(n)}(R)\}_{n\geq 1}$, such that $\lim_{n\rightarrow \infty} P_b^{(n)} = 0$ and $\lim_{n\rightarrow \infty} R_n = R$.  We then say that the sequence of codes $\{\mathcal{C}^{(n)}(R)\}_{n\geq 1}$ \textit{achieves} a rate $R$ over the $(d,\infty)$-RLL input-constrained channel.
	The capacity, $C_{(d,\infty)}$, is defined to be the supremum over the respective $(d,\infty)$-achievable rates, and is a function of the parameters of the noise process. 
	Finally, a family of sequences of codes $\{\{\hat{\mathcal{C}}^{(n)}_{\mathbf{p}}\}_{n\geq 1}\}$, indexed by the noise parameters $\mathbf{p}$, is said to be \emph{capacity-achieving} (or $(d,\infty)$-capacity-achieving), under bit-MAP decoding, if for all $\mathbf{p}$, $\{\hat{\mathcal{C}}^{(n)}_\mathbf{p}\}_{n\geq 1}$ achieves any rate $R\in (0,C_{(d,\infty)}(\mathbf{p}))$ over the $(d,\infty)$-RLL input-constrained channel. Similar definitions hold under block-MAP decoding, as well. Note that the definitions also hold when $d=0$, which represents the unconstrained channel.
\end{definition}

%
%
%
\section{Main Results}
\label{sec:main}
Before we state our upper bound on the rates of linear RLL subcodes of RM codes, we recall the result of Reeves and Pfister in \cite{Reeves}, which provides context to our using RM codes over input-constrained BMS channels. 
For a given $R\in (0,1)$, consider any sequence of RM codes $\{\mathcal{C}_m(R) = \text{RM}(m,r_m)\}_{m\geq 1}$, under the lexicographic ordering of coordinates, with $R_m$ being the rate of $\mathcal{C}_m(R)$, such that $R_m\to R$ as $m\to \infty$. 
The following theorem then holds true:
%

\begin{theorem}[Theorem 1 of \cite{Reeves}]
	\label{thm:Reeves}
	Consider an unconstrained BMS channel with capacity $C\in (0,1)$. Then, any rate $R\in [0,C)$ is achieved by the sequence of codes $\{\mathcal{C}_m(R)\}_{m\geq 1}$, under bit-MAP decoding.
\end{theorem}
Hence, the families of codes described above 
are $(0,\infty)$-capacity-achieving, under bit-MAP decoding. 

We now discuss a theorem that provides upper bounds on the largest rate achievable, using linear subcodes of RM codes, over a $(d,\infty)$-RLL input-constrained BMS channel. Fix any sequence of codes $\{{\mathcal{C}}_m(R) = \text{RM}(m, r_m)\}_{m\geq 1}$, which achieves a rate $R$ over the unconstrained BMS channel. Let $\overline{\mathcal{C}}_{d}^{(m)}$ denote the largest \emph{linear} subcode of ${\mathcal{C}}_m(R)$, all of whose codewords respect the $(d,\infty)$-RLL constraint. We then define
\begin{equation}
	\label{eq:Rub}
\mathsf{R}^{(d,\infty)}_{{\mathcal{C}},\text{Lin}}(R):=\limsup_{m\to \infty}\frac{\log_2\left \lvert \overline{\mathcal{C}}_{d}^{(m)}\right\rvert}{2^m},
\end{equation}
to be the largest rate achieved by linear $(d,\infty)$-RLL subcodes of $\{{C}_m(R)\}$, assuming that the ordering of the coordinates of the code is according to the lexicographic ordering. Then,

\begin{theorem}
	\label{thm:rmlinub}
	For any sequence of codes $\{{\mathcal{C}}_m(R) = \text{RM}(m,r_m)\}_{m\geq 1}$, with rate$(\mathcal{C}_m(R))\xrightarrow{m\to \infty} R$, it holds that
	\[
	\mathsf{R}^{(d,\infty)}_{{\mathcal{C}},\text{Lin}}(R) \leq \frac{R}{d+1}.
	\]
\end{theorem}
Hence, from Theorem \ref{thm:Reeves}, the largest rate achievable over a $(d,\infty)$-RLL input-constrained BMS channel, under bit-MAP decoding, using linear $(d,\infty)$-RLL subcodes of RM codes, is bounded above by $\frac{C}{d+1}$, where $C$ is the capacity of the unconstrained BMS channel. Theorem \ref{thm:rmlinub} is proved in Section \ref{sec:rmublin}. Now, consider the sequence of RM codes $\{\hat{\mathcal{C}}_m(R) = \text{RM}(m,v_m)\}_{m\geq 1}$, with
\begin{equation}
	\label{eq:rmval}
	v_m = \max \left\{\left \lfloor \frac{m}{2}+\frac{\sqrt{m}}{2}Q^{-1}(1-R)\right \rfloor,0\right\},
\end{equation}
where $Q(\cdot)$ is the complementary cumulative distribution function (c.c.d.f.) of the standard normal distribution. 

Now, for a fixed $d\geq 1$, let $z := \left \lceil \log_2(d+1)\right \rceil$. Consider the subcode $\mathcal{C}_m^{(d,\infty)}(R)$, of the code $\mathcal{C}_m(R)$, defined as:
\begin{align}
	\label{eq:rmlb1}
	\mathcal{C}_m^{(d,\infty)}(R):=\Bigg\{\text{Eval}(f): f = &\bigg(\prod_{i=m-z+1}^{m} x_i \bigg)\cdot g(x_{1},\ldots, x_{m-z}),\notag\\ &\text{ where } \text{deg}(g)\leq r_m-z\Bigg\}.
\end{align}
 Note that $\mathcal{C}_m^{(d,\infty)}(R)$ is a \emph{linear} subcode of $\mathcal{C}_m(R)$. The following theorem from \cite{arnk22arxiv} then holds:
\begin{theorem}[Theorem III.2 in \cite{arnk22arxiv}]
	\label{thm:rm}
	For any $R \in (0,C)$, the sequence of linear codes $\{{\mathcal{C}}_{m}^{(d,\infty)}(R)\}_{m\geq 1}$, where ${\mathcal{C}}_{m}^{(d,\infty)}(R) \subset {\mathcal{C}}_m(R)$, achieves a rate of $\frac{R}{2^{\left \lceil \log_2(d+1)\right \rceil}}$, over a $(d,\infty)$-RLL input-constrained BMS channel, under bit-MAP decoding.
\end{theorem}

 Thus, Theorem \ref{thm:rmlinub} shows that the sequence of linear subcodes $\{{\mathcal{C}}_{m}^{(d,\infty)}(R)\}_{m\geq 1}$, in equation \eqref{eq:rmlb1}, is rate-optimal whenever $d+1$ is a power of $2$, in that it achieves the rate upper bound of $R/(d+1)$. We remark here that the problem of identifying linear codes that are subsets of the set of $(d,\infty)$-RLL sequences of a fixed length, has been studied \cite{lechner}. The results therein show that the largest linear code within $S_{(d,\infty)}^{(m)}$ has rate no larger than $\frac{1}{d+1}$, as $m\to \infty$. However, such a result offers no insight into rates achievable over BMS channels.
We then consider situations where the coordinates of the RM codes follow orderings different from the standard lexicographic ordering. First, we study upper bounds on the rates of linear $(d,\infty)$-RLL subcodes of RM codes, ordered according to a Gray ordering (see Section \ref{sec:perm} for a description of a Gray ordering). For a fixed $R\in (0,C)$, let $\{{\mathcal{C}}_m^\text{G}(R)\}_{m\geq 1}$ be any sequence of RM codes under a Gray ordering, such that rate$(\mathcal{C}_m^\text{G}(R))\xrightarrow{m\to \infty} R$. Further, for every $m$, let $\overline{\mathcal{C}}_{d,\text{G}}^{(m)}$ be the largest \emph{linear} subcode of ${\mathcal{C}}_m^\text{G}(R)$. We also define
\begin{equation}
	\label{eq:Rubgray}
	\mathsf{R}^{(d,\infty)}_{{\mathcal{C}^\text{G}},\text{Lin}}(R):=\limsup_{m\to \infty}\frac{\log_2\left \lvert \overline{\mathcal{C}}_{d,\text{G}}^{(m)}\right\rvert}{2^m}
\end{equation}
to be the largest rate achieved by linear $(d,\infty)$-RLL subcodes of $\{\mathcal{C}_m^\text{G}(R)\}_{m\geq 1}$.
We obtain the following result:

\begin{theorem}
	\label{thm:grayinf}
	For any sequence of RM codes under a Gray ordering, $\{{\mathcal{C}}_m^\text{G}(R)\}_{m\geq 1}$, with rate$(\mathcal{C}_m^\text{G}(R))\xrightarrow{m\to \infty} R$, it holds that 
	\[
	\mathsf{R}^{(d,\infty)}_{{\mathcal{C}^\text{G}},\text{Lin}}(R)\leq \frac{R}{d+1}.
	\]
\end{theorem}
The proof of Theorem \ref{thm:grayinf} is provided in Section \ref{sec:perm}.

Now, we consider arbitrary orderings of coordinates, defined by the sequence of permutations $(\pi_m)_{m\geq 1}$, with $\pi_m: [0:2^m-1]\to [0:2^m-1]$. As with the Gray ordering, we define the sequence of $\pi$-ordered RM codes $\{\mathcal{C}_m^\pi(R)\}_{m\geq 1}$, with
\begin{align*}\mathcal{C}_{m}^\pi(R):= \big\{(c_{\pi_m(0)},c_{\pi_m(2)}&,\ldots,c_{\pi_m(N_m-1)}):\\ &(c_0,c_1,\ldots,c_{N_m-1})\in {\mathcal{C}}_m(R)\big\}.\end{align*}
We also define $\overline{\mathcal{C}}_{d,\pi}^{(m)}$ be the largest \emph{linear} $(d,\infty)$-RLL subcode of $\mathcal{C}_m^\pi(R)$. The theorem below is then shown to hold:
\begin{theorem}
	\label{thm:genperminf}
	For large $m$ and for all but a vanishing fraction of coordinate permutations, $\pi_m: [0:2^m-1]\to [0:2^m-1]$, the following rate upper bound holds:
	\[
	\frac{\log_2\left \lvert \overline{\mathcal{C}}_{d,\pi}^{(m)}\right\rvert}{2^m}\leq \frac{R}{d+1}+\delta_m,
	\]
	where $\delta_m\xrightarrow{m\to \infty} 0$.
\end{theorem}
Section \ref{sec:perm} contains the proof of Theorem \ref{thm:genperminf}.

Next, we turn our attention to the design of non-linear $(d,\infty)$-RLL codes, whose rates improve on those in Theorem \ref{thm:rm}. Our next theorem, stated below informally, uses cosets of RM codes, for this purpose. We denote by $C_0^{(d)}$, the noiseless capacity of the $(d,\infty)$-RLL constraint, and by $C$, the capacity of the unconstrained BMS channel. 

\begin{theorem}[Informal]
	\label{thm:rmcosets}
	For any BMS channel of capacity $C$, there exists a sequence of $(d,\infty)$-RLL constrained codes $\{{\mathcal{C}}_m^{\text{cos}}\}_{m\geq 1}$, using cosets of RM codes, such that
	\[
	\liminf_{m\to \infty} \text{rate}(\mathcal{C}_m^{\text{cos}})\geq \frac{C_0^{(d)}\cdot C^2\cdot 2^{-\left \lceil \log_2(d+1)\right \rceil}}{C^2\cdot 2^{-\left \lceil \log_2(d+1)\right \rceil} + 1-C+2^{-\tau}},
	\]
with the above bound being achievable over any $(d,\infty)$-RLL input-constrained BMS channel. Here, $\tau$ is an arbitrarily large, but fixed, positive integer.
\end{theorem}
It can be checked that the rates achieved using Theorem \ref{thm:rmcosets} are better than those achieved using Theorem \ref{thm:rm} (and in fact, better than those achieved using any sequence of linear $(d,\infty)$-RLL subcodes of RM codes), for low noise regimes of the BMS channel. For example, when $d=1$, the rates achieved using the codes in Theorem \ref{thm:rmcosets} are better than those achieved using linear subcodes, for certain values of $C\gtrapprox 0.7613$. Figures \ref{fig:first} and \ref{fig:second} show comparisons between the lower bounds (achievable rates) in Theorems \ref{thm:rm} and \ref{thm:rmcosets}, with the coset-averaging bound of \cite{pvk}, for $d=1$ and $d=2$, respectively. While \cite{pvk} provides existence results on rates achieved using cosets of RM codes, with the rates calculated therein being better than those in Theorem \ref{thm:rmcosets} in the low noise regimes of the BMS channel, our construction is more explicit. A discussion on the construction leading to Theorem \ref{thm:rmcosets} is taken up in Section \ref{sec:cosets}.
\begin{figure*}%
	\centering	
		\includegraphics[width=0.8\textwidth]{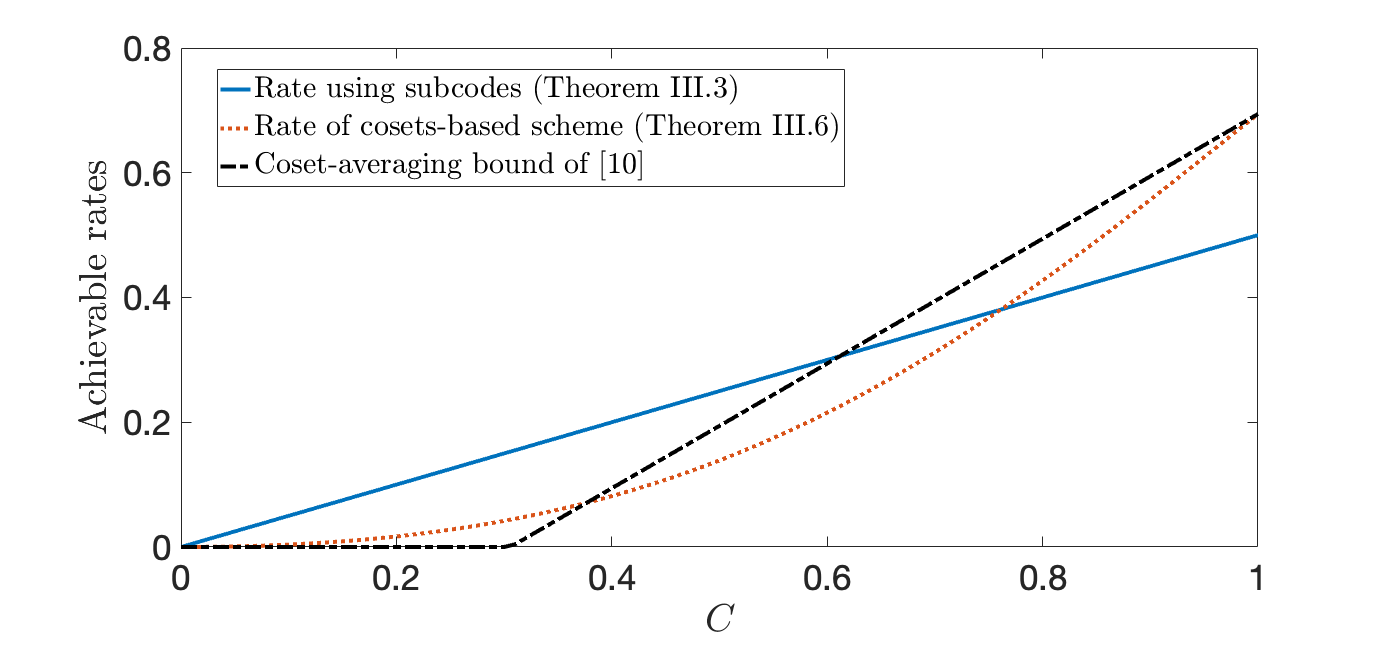}%
	\caption{Plot comparing, for $d=1$, the rate lower bound of $C/2$ achieved using subcodes, from Theorem \ref{thm:rm}, the rate lower bound achieved using Theorem \ref{thm:rmcosets}, with $\tau = 50$, and the coset-averaging lower bound of $\max(0,C_0^{(1)}+C-1)$, of \cite{pvk}. Here, the noiseless capacity, $C_0^{(1)} \approx 0.6942$.}
	\label{fig:first}%
\end{figure*}

%

We end this section with a remark. Note that the all-ones codeword $\mathbf{1}$ belongs to any RM code. Since any codeword $\mathbf{c}$ that respects the $(0,1)$-RLL constraint can be written as $\mathbf{c} = \mathbf{1}+\mathbf{\hat{c}}$, where $\mathbf{\hat{c}}$ respects the $(1,\infty)$-RLL constraint, the lower bound of Theorem \ref{thm:rm} and the upper bound of Theorem \ref{thm:rmlinub} hold for the rates of $(0,1)$-RLL subcodes as well. 


\section{Upper Bounds for Linear Subcodes}
\label{sec:rmublin}
In this section, we derive upper bounds on the rates achieved by linear $(d,\infty)$-RLL subcodes of any sequence of RM codes of rate $R$. We fix a sequence of codes $\{{\mathcal{C}_m(R)} = \text{RM}(m,r_m)\}$ that achieves a rate $R$ over the unconstrained BMS channel. 

We first state and prove a fairly general proposition on the rates of linear $(d,\infty)$-RLL subcodes of linear codes. Recall that for a linear code $\mathcal{C}$ over $\mathbb{F}_2$, of blocklength $N$ and dimension $K$, an information set is a collection of $K$ coordinates in which all possible $K$-tuples over $\mathbb{F}_2$ can appear. Equivalently, if $G$ is any generator matrix for $\mathcal{C}$, an information set is a set of $K$ column indices such that $G$ restricted to those columns is a full-rank matrix. 

\begin{proposition}
	\label{prop:linub}
	Let $\overline{\mathcal{C}}$ be an $[N,K]$ binary linear code. If $\mathcal{I}$ is an information set of $\overline{\mathcal{C}}$ that contains $t$ disjoint $(d+1)$-tuples of consecutive coordinates $(i_1, i_1 + 1,\ldots,i_1+d), (i_2, i_2 + 1,\ldots,i_2+d), ..., (i_t, i_t + 1,\ldots,i_t+d)$, with $i_1\geq 1$, $i_j>i_{j-1}+d$, for all $j\in [2:t]$, and $i_t\leq n-d$, then the dimension of any linear $(d,\infty)$-RLL subcode of $\overline{\mathcal{C}}$ is at most $K-dt$. 
\end{proposition}
\begin{proof}
	Suppose that the information set $\mathcal{I}$ contains exactly $t$ disjoint $(d+1)$-tuples of consecutive coordinates as in the statement of the proposition. By definition, all possible $K$-tuples appear in the coordinates in $\mathcal{I}$. Now, consider any linear $(d,\infty)$-RLL subcode of $\overline{\mathcal{C}}$, and any $(d+1)$-tuple of consecutive coordinates $\{i_j,i_j+1,\ldots,i_j+d\} \in \mathcal{I}$, for $j\in [t]$. Since the $(d,\infty)$-RLL constraint requires that successive $1$s be separated by at least $d$ $0$s (and by linearity of the subcode), the only possible tuples of $d+1$ consecutive symbols, in any codeword in the subcode, are $(0,0,\ldots,0)$ and one of $\mathbf{e}_i^{(d+1)}$, for $i\in [d+1]$. This is because, if $\mathbf{e}_i^{(d+1)}$ and $\mathbf{e}_j^{(d+1)}$ both occur in a collection of $d+1$ consecutive positions, then, by linearity of the subcode $\overline{\mathcal{C}}$, it holds that $\mathbf{e}_i^{(d+1)}+\mathbf{e}_j^{(d+1)}$ (where the addition is over vectors in $\mathbb{F}_2^{d+1}$) must occur in some codeword of the subcode, thereby making the codeword not $(d,\infty)$-RLL. Hence, for every $(d+1)$-tuple of consecutive coordinates, only a $2^{-d}$ fraction of the $2^{d+1}$ possible tuples are allowed. Thus, overall, the number of codewords in the linear $(d,\infty)$-RLL subcode is at most $\frac{2^K}{2^{dt}}$. The result then follows straightforwardly.
	
\end{proof}
In order to obtain an upper bound, as in Theorem \ref{thm:rmlinub}, on the rate of linear $(d,\infty)$-RLL subcodes of the sequence of codes $\{{\mathcal{C}_m(R)}\}_{m\geq 1}$, we shall first identify an information set $\mathcal{I}_{m,r_m}$ of ${\mathcal{C}_m(R)} = \text{RM}(m,r_m)$. We then compute the number of disjoint $(d+1)$-tuples of consecutive coordinates in $\mathcal{I}_{m,r_m}$, and apply Proposition \ref{prop:linub} to get an upper bound on the dimension of the linear constrained subcodes.

We introduce some notation for ease of reading: given a matrix $M_{p\times q}$, we use the notation $M[\mathcal{U},\mathcal{V}]$ to denote the submatrix of $M$ consisting of the rows in the set $\mathcal{U}\subseteq [p]$ and the columns in the set $\mathcal{V}\subseteq [q]$. We also recall the definition of the generator matrix $G_{\text{Lex}}(m,r)$, of RM$(m,r)$, and the indexing of columns of the matrix, from Section \ref{sec:introrm}. We also interchangeably index the coordinates of a codeword of RM$(m,r)$ by integers $i\in [0,2^m-1]$, and by $m$-tuples of binary symbols. Further, the notation $\mathbf{e}_\mathbf{b}^{(2^m)}$ denotes the standard basis vector with a $1$ in the coordinate indexed by $\mathbf{b} = (b_1,\ldots,b_m)$, in the lexicographic order. The superscript `$(2^m)$' will be dropped when clear from the context. 

Now, given the code RM$(m,r)$, consider the binary linear code (a subspace of $\mathbb{F}_2^{2^m}$), $\tilde{\mathcal{C}}({m,r})$, spanned by the codewords in the set 

\begin{equation}
	\label{eq:Bmr}
	\mathcal{B}_{m,r}:=\left\{\text{Eval}\left(\prod_{i\in S}x_i\right): S\subseteq [m]\ \text{with } |S|\geq r+1\right\}.
\end{equation}

It can be checked that the vectors in $\mathcal{B}_{m,r}$ are also linearly independent, and, hence, $\mathcal{B}_{m,r}$ forms a basis for $\tilde{\mathcal{C}}({m,r})$, with dim$\left(\tilde{\mathcal{C}}({m,r})\right) = {m \choose \geq r+1}$. Moreover, the codewords in $\tilde{\mathcal{C}}({m,r})$ are linearly independent from codewords in RM$(m,r)$, by definition.

The following lemma identifies an alternative basis for $\tilde{\mathcal{C}}({m,r})$, which will prove useful in our analysis, later on.

\begin{lemma}
	\label{lem:quotient}
	Consider the code $\tilde{\mathcal{C}}({m,r})=\text{span}\left(\mathcal{B}_{m,r}\right)$, where $\mathcal{B}_{m,r}$ is as in \eqref{eq:Bmr}. It holds that $\tilde{\mathcal{C}}({m,r}) = \text{span}\left(\{\mathbf{e}_{\mathbf{b}}: \text{wt}(\mathbf{b})\geq r+1\}\right)$.
	\end{lemma}
\begin{proof}
	
	Note that any standard basis vector $\mathbf{e}_{\mathbf{b}}$, with  wt$(\mathbf{b})\geq r+1$, can be written as Eval$(f)$, where 
	\[
	f(x_1,\ldots,x_m) = \prod_{i\in \text{supp}(\mathbf{b})}x_i \cdot \prod_{i\notin \text{supp}(\mathbf{b})}(1+x_j).
	\]
	From the fact that wt$(\mathbf{b})\geq r+1$, it holds that the degree of any monomial in $f$ is at least $r+1$, and hence, Eval$(f) = \mathbf{e}_\mathbf{b} \in \text{span}(\mathcal{B}_{m,r}) = \tilde{\mathcal{C}}({m,r})$. The result follows by noting that $\{\mathbf{e}_{\mathbf{b}}: \text{wt}(\mathbf{b})\geq r+1\}$ is a collection of linearly independent vectors, of size ${m\choose \geq r+1}$, which, in turn, equals dim$\left(\tilde{\mathcal{C}}({m,r})\right)$.
\end{proof}
\begin{lemma}
	\label{lem:infoset}
	An information set of $\text{RM}(m,r)$ is the set of coordinates $\mathcal{I}_{m,r}:= \{\mathbf{b} = (b_1,\ldots,b_m)\in \mathbb{F}_2^m: \text{wt}(\mathbf{b})\leq r\}$.
\end{lemma}
\begin{proof}
	 In order to prove that $\mathcal{I}_{m,r}$ is an information set of RM$(m,r)$, it is sufficient to show that $G_{\text{Lex}}(m,r)$ restricted to the columns in $\mathcal{I}_{m,r}$ is of full rank.
	
	Now, consider the generator matrix $\tilde{G}(m,r)$, of $\tilde{\mathcal{C}}(m,r)$, consisting of rows that are vectors in $\mathcal{B}_{m,r}$. We build the $2^m\times 2^m$ matrix $$\mathsf{H}:= \begin{bmatrix} \begin{array}{c}
		\tilde{G}(m,r)\\
		\hline\\
		G_{\text{Lex}}(m,r)
		\end{array}
	\end{bmatrix},$$
with $\mathsf{H}$ being full rank. Note that, from Lemma \ref{lem:quotient}, any standard basis vector $\mathbf{e}_{\mathbf{b}}$, with $\mathbf{b} \in \mathcal{I}_{m,r}^c$, belongs to rowspace$(\tilde{G}(m,r))$. By Gaussian elimination, it is then possible to replace the first ${m\choose \geq r+1}$ rows of $\mathsf{H}$, corresponding to the submatrix $\tilde{G}(m,r)$, with the standard basis vectors $\mathbf{e}_\mathbf{b}$, with $\mathbf{b}\in \mathcal{I}_{m,r}^c$. Clearly, from the fact that $\mathsf{H}$ is full rank, this then means that $\mathsf{H}\left[\left[{m\choose \geq r+1}+1:2^m\right],\mathcal{I}_{m,r}\right]$ is full rank, or, $G_{\text{Lex}}(m,r)$, restricted to columns in $\mathcal{I}_{m,r}$, is full rank.
\end{proof}

Now that we have identified an information set $\mathcal{I}_{m,r_m}$ of $\mathcal{C}_m(R) = $ RM$(m,r_m)$, we need only calculate the number of disjoint $(d+1)$-tuples of consecutive coordinates in $\mathcal{I}_{m,r_m}$. We introduce the notation \textbf{B}$(i)$ to denote the length-$m$ binary representation of $i$, for $0\leq i\leq 2^m-1$. We also define a ``run'' of coordinates belonging to a set $\mathcal{A}\in \{0,1\}^m$, to be a contiguous collection of coordinates, $\left(i,i+1,\ldots,i+\ell\right)$, such that $\mathbf{B}(j)\in \mathcal{A}$, for all $i\leq j\leq i+\ell$, and $\mathbf{B}(i-1),\mathbf{B}(i+\ell+1)\notin \mathcal{A}$, where $i \in [0:2^m-1-\ell]$. Further, the length of such a run of coordinates is exactly $\ell$. 

We shall first compute the number of runs of consecutive coordinates, in the lexicographic ordering, which belong to the information set $\mathcal{I}_{m,r_m}$. Formally, if we define

\begin{figure*}%
	\centering

	\includegraphics[width=0.8\textwidth]{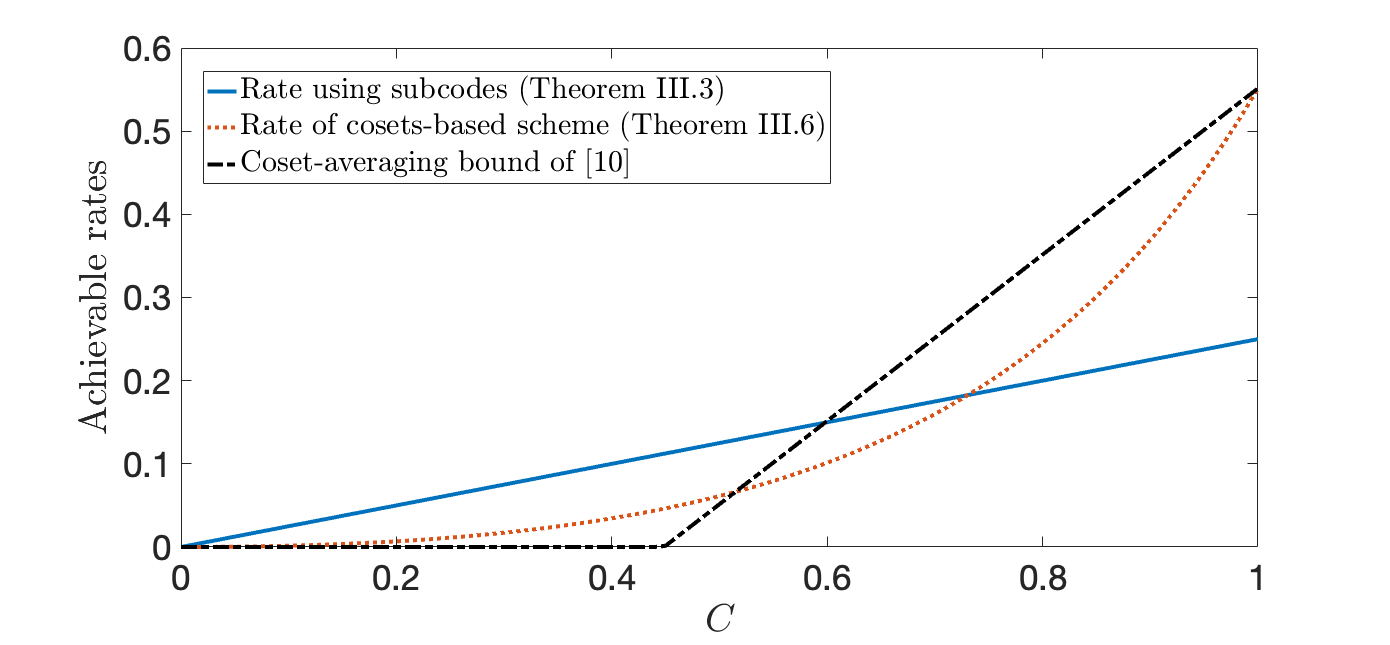}%
	\caption{Plot comparing, for $d=2$, the rate lower bound of $C/4$ achieved using subcodes, from Theorem \ref{thm:rm}, the rate lower bound achieved using Theorem \ref{thm:rmcosets}, with $\tau = 50$, and the coset-averaging lower bound of $\max(0,C_0^{(2)}+C-1)$, of \cite{pvk}. Here, the noiseless capacity, $C_0^{(2)} \approx 0.5515$.}
	\label{fig:second}%
\end{figure*}

\begin{align}
	\label{eq:gammam}
\Gamma_{m,r_m} := \{s: \textbf{B}(s+1)\notin \mathcal{I}_{m,r_m},\ &\text{and}\ \textbf{B}(s-p),\ldots,\textbf{B}(s)\in \mathcal{I}_{m,r_m}, \notag\\&\text{ for some $p\geq 0$}\},
\end{align}
to be the set of right end-point coordinates of runs that belong to $\mathcal{I}_{m,r_m}$, then the required number of runs is $\left \lvert \Gamma_{m,r_m}\right \rvert$.
\begin{lemma}
	\label{lem:runs}
	Under the lexicographic ordering, it holds that $\left \lvert \Gamma_{m,r}\right \rvert={m-1\choose r}$, for $0\leq r\leq m-1$.
\end{lemma}
\begin{proof}
	Let $r\in [0:m-1]$. Note that every right end-point of a run, $s\in \Gamma_{m,r}$,  with $s\in [0:2^m-2]$, is such that wt$(\mathbf{B}(s))\leq r$, but wt$(\mathbf{B}(s+1))\geq r+1$. We now claim that an integer $s \in \Gamma_{m,r}$ iff \textbf{B}$(s)=(b_1,\ldots ,b_{m-1},0)$, for $b_1,\ldots,b_{m-1}\in \{0,1\}$, with wt$((b_1,\ldots ,b_{m-1},0)) = r$. 
	
	To see this, note that if \textbf{B}$(s)=(b_1,\ldots, b_{m-1},0)$, then \textbf{B}$(s+1) = (b_1,\ldots, b_{m-1},1)$. Hence, if wt$((b_1,\ldots, b_{m-1})) = r$, then $s\in \Gamma_{m,r}$. Conversely, if $s\in \Gamma_{m,r}$, then \textbf{B}$(s)$ cannot end in a $1$. Indeed, if this were the case, then we would have \textbf{B}$(s)$ being of the form $(b_1,\ldots, b_\ell, 0,1, \ldots, 1)$, with $b_1,\ldots,b_{\ell}\in \{0,1\}$, so that \textbf{B}$(s+1)$ would be $(b_1,\ldots, b_\ell, 1,0, \ldots, 0)$, the weight of which does not exceed that of \textbf{B}$(s)$. So, \textbf{B}$(s)$ must be of the form $(b_1, \ldots, b_{m-1}, 0)$, and so, \textbf{B}$(s+1) = (b_1,\ldots, b_{m-1},1)$. From wt$(\mathbf{B}(s)) \le r$ and wt$(\mathbf{B}(s+1)) \ge r+1$, we obtain that wt$(b_1 \ldots b_{m-1}) = r$. 
	
	This then implies that the number of runs, which is equal to the number of right end-points of runs, exactly equals ${m-1\choose r}$.
\end{proof}
With the ingredients in place, we are now in a position to prove Theorem \ref{thm:rmlinub}.
\begin{proof}[Proof of Theorem \ref{thm:rmlinub}]
	Fix a sequence of codes $\{{\mathcal{C}_m}(R) = \text{RM}(m,r_m)\}_{m\geq 1}$ that achieves a rate $R\in (0,1)$ over the unconstrained BMS channel, with $r_m\leq m-1$, for all $m$. We use the notation $K_m:={m\choose \leq r_m}$ to denote the dimension of ${{\mathcal{C}_m(R)}}$. 
		
	Now, for a given $m$, consider the information set $\mathcal{I}_{m,r_m}$ as in Lemma \ref{lem:infoset}. We know from Lemma \ref{lem:runs} that the number of runs under the lexicographic ordering, $\left \lvert \Gamma_{m,r_m}\right \rvert$, of coordinates that lie in $\mathcal{I}_{m,r_m}$, is exactly ${m-1\choose r_m}$. Now, note that the $i^{\text{th}}$ run $(s_i,\ldots,s_i+\ell_i)$, of length $\ell_i$, with $s_i\in \Gamma_{m,r_m}$ and $i\in \left[\left \lvert \Gamma_{m,r_m}\right \rvert\right]$, contributes $\left \lfloor \frac{\ell_i}{d+1}\right \rfloor$ disjoint $(d+1)$-tuples of consecutive coordinates in $\mathcal{I}_{m,r}$. It then holds that the overall number of disjoint $(d+1)$-tuples of consecutive coordinates in $\mathcal{I}_{m,r}$ is $t_m$, where
	\begin{align*}
		t_m &= \sum_{i=1}^{\left \lvert \Gamma_{m,r_m}\right \rvert}\left \lfloor \frac{\ell_i}{d+1} \right \rfloor\\
		&\geq \sum_{i=1}^{\left \lvert \Gamma_{m,r_m}\right \rvert} \left(\frac{\ell_i}{d+1} -1\right)\\
		&= \frac{K_m}{d+1} - \left \lvert \Gamma_{m,r_m} \right \rvert
		= \frac{K_m}{d+1} - {m-1 \choose r_m},
	\end{align*}
where the last equality follows from Lemma \ref{lem:runs}.
	
	Using Proposition \ref{prop:linub}, it follows that the dimension of any linear $(d,\infty)$-RLL subcode of ${\mathcal{C}_m(R)}$ is at most $K_m-dt_m$. It then holds that
	\begin{align*}
		\mathsf{R}^{(d,\infty)}_{{\mathcal{C}},\text{Lin}}(R)&=\limsup_{m\to \infty}\frac{\log_2\left \lvert \overline{\mathcal{C}}_{d}^{(m)}\right\rvert}{2^m}\\
		&\leq \limsup_{m\to \infty}\frac{K_m-dt_m}{2^m}\\
		&\leq \limsup_{m\to \infty}\frac{K_m-\frac{dK_m}{d+1}+d\cdot {m-1\choose r_m}}{2^m}\\
		&\leq \lim_{m\to \infty}\frac{\frac{K_m}{d+1}+d\cdot {m-1\choose \left \lfloor \frac{m-1}{2}\right \rfloor}}{2^m}\\
		&= \frac{R}{d+1},
	\end{align*}
	where the last equality holds from the fact that ${m-1\choose \left \lfloor \frac{m-1}{2}\right \rfloor} \sim c\cdot \frac{2^m}{\sqrt{m-1}}$ (see. for example, equation $(5.28)$ in \cite{asymp}, where `$\sim$' is used to mean ``grows as"), and $\lim_{m\to \infty} \frac{K_m}{2^m} = R$.
\end{proof}
\section{Alternative Coordinate Orderings}
\label{sec:perm}
Throughout the previous sections, we have assumed that the coordinates of the Reed-Muller code are ordered according to the standard lexicographic ordering. Since permutations of coordinates have the potential to convert a binary word that does not respect the $(d,\infty)$-RLL constraint to one that does, we ask the question if under alternative coordinate orderings, we can obtain linear $(d,\infty)$-RLL subcodes of RM codes, of rate larger than the upper bound in Theorem \ref{thm:rmlinub}.

First, we consider a Gray ordering of coordinates of the code RM$(m,r)$. In such an ordering, consecutive coordinates $\mathbf{b} = (b_1,\ldots,b_m)$ and $\mathbf{b}^\prime = (b_1^\prime,\ldots,b_m^\prime)$ are such that for some bit index $i\in [m]$, $b_i\neq b_i^\prime$, but $b_j = b_j^\prime$, for all $j\neq i$. In words, consecutive coordinates in a Gray ordering, when represented as $m$-tuples, differ in exactly one bit index. Note that multiple orderings are possible, which satisfy this property. We remark that any fixed Gray ordering can also be seen as a Hamiltonian path (see, for example, \cite{diestel}, Chap. 10) on the $m$-dimensional unit hypercube. 

In what follows, we work with a fixed sequence of Gray orderings defined as follows: let $(\pi_m^\text{G})_{m\geq 1}$ be a sequence of permutations, with $\pi_m^\text{G}: [0:2^m-1]\rightarrow [0:2^m-1]$, for any $m\geq 1$, having the property that \textbf{B}$(\pi_m^\text{G}(j))$ differs from \textbf{B}$(\pi_m^\text{G}(j-1))$ in exactly one bit index, for any $j\in [0:2^m-1]$. Here, again, \textbf{B}$(z)$ is the $m$-length binary representation of $z$, for $z\in [0:2^m-1]$. 

Now, fix a sequence of codes $\{{\mathcal{C}_m}(R) = \text{RM}(m,r_m)\}_{m\geq 1}$ that achieves a rate $R\in (0,1)$ over the unconstrained BMS channel, with $r_m\leq m-1$, for all $m$. We again use the notation $K_m:={m\choose \leq r_m}$ to denote the dimension of ${{\mathcal{C}_m(R)}}$. We then define the sequence of Gray-ordered RM codes $\{\mathcal{C}_m^{\text{G}}(R)\}_{m\geq 1}$, with
\begin{align*}\mathcal{C}_{m}^\text{G}(R):= \big\{(c_{\pi_m^\text{G}(0)},c_{\pi_m^\text{G}(1)}&,\ldots,c_{\pi_m^\text{G}(2^m-1)}):\\ &(c_0,c_1,\ldots,c_{2^m-1})\in {\mathcal{C}}_m(R)\big\}.\end{align*}
Clearly, the sequence of codes $\{\mathcal{C}_m^{\text{G}}(R)\}_{m\geq 1}$ also achieves a rate $R\in (0,1)$ over the unconstrained BMS channel. In order to obtain an upper bound on the rate of the largest linear $(d,\infty)$-RLL subcode of the code $\mathcal{C}_m^{\text{G}}(R)$, as in Section \ref{sec:rmublin}, we shall work with the same information set $\mathcal{I}_{m,r_m}$ as in Lemma \ref{lem:infoset}. Note that the coordinates of the Gray-ordered RM code are now represented by $m$-tuples, in which the $j^\text{th}$ coordinate from the beginning is \textbf{B}$(\pi_m^\text{G}(j-1))$.

Again, we define the set
\begin{align*}
\Gamma_{m,r_m}^\text{G}:=\{\pi_m^\text{G}(s):\ &\textbf{B}(\pi_m^\text{G}(s+1))\notin \mathcal{I}_{m,r_m},\ \textbf{B}(\pi_m^\text{G}(s-p)),\ldots,\\&\textbf{B}(\pi_m^\text{G}(s))\in \mathcal{I}_{m,r_m},\text{ for some $p\geq 0$}\},
\end{align*}
to be set of right end-point coordinates of runs that belong to $\mathcal{I}_{m,r_m}$, with the number of such runs being $\left \lvert \Gamma_{m,r_m}^\text{G}\right \rvert$.

We now state and prove a lemma analogous to Lemma \ref{lem:runs}:
\begin{lemma}
	\label{lem:runsgray}
	Under a fixed Gray ordering defined by $\pi_m^\text{G}$, it holds that $\left \lvert\Gamma_{m,r_m}^\text{G} \right \rvert \leq {m\choose r_m+1}$, for $0\leq r_m\leq m-1$.
\end{lemma}
\begin{proof}
	As before, every run of coordinates that belong to $\mathcal{I}_{m,r_m}$ has a right end point, which is an integer $\pi_m^\text{G}(s)\in [0:2^m-2]$, such that wt$(\mathbf{B}(\pi_m^\text{G}(s)))\leq r_m$, but wt$(\mathbf{B}(\pi_m(s+1)))\geq r_m+1$. Now, under the Gray ordering, since consecutive coordinates differ in exactly one bit index, it can be seen that an integer $\pi_m^\text{G}(s) \in \Gamma_{m,r_m}^\text{G}$, only if wt$(\mathbf{B}(\pi_m^\text{G}(s+1)))= r_m+1$. Thus, the number of runs is bounded above by ${m\choose r_m+1}$, which is the number of appearances of coordinates whose binary representation has weight exactly $r_m+1$.
\end{proof}

With Lemma \ref{lem:runsgray} established, we now embark on a proof of Theorem \ref{thm:grayinf}.
\begin{proof}[Proof of Theorem \ref{thm:grayinf}]
	Similar to the proof of Theorem \ref{thm:rmlinub}, the calculation of the overall number, $t_m^\text{G}$, of disjoint $(d+1)$-tuples of consecutive coordinates in $\mathcal{I}_{m,r_m}$, results in
	\begin{align*}
		t_m^\text{G} \geq \frac{K_m}{d+1}-{m\choose r_m+1}.
	\end{align*}
	Again, using Proposition \ref{prop:linub}, it follows that the dimension of any linear $(d,\infty)$-RLL subcode of ${\mathcal{C}_m^\text{G}(R)}$ is at most $K_m-dt_m^\text{G}$. Now, we recall the definition of $\mathsf{R}^{(d,\infty)}_{{\mathcal{C^\text{G}}},\text{Lin}}(R)$, from equation \eqref{eq:Rubgray}. It then holds that $\mathsf{R}^{(d,\infty)}_{{\mathcal{C^\text{G}}},\text{Lin}}(R)$ obeys:
	\begin{align*}
		\mathsf{R}^{(d,\infty)}_{{\mathcal{C^\text{G}}},\text{Lin}}(R)&\leq \limsup_{m\to \infty}\frac{K_m-dt_m^\text{G}}{2^m}\\
		&\leq \limsup_{m\to \infty}\frac{K_m-\frac{dK_m}{d+1}+d\cdot {m\choose r_m+1}}{2^m}\\
		&\leq \lim_{m\to \infty}\frac{\frac{K_m}{d+1}+d\cdot {m\choose \left \lfloor \frac{m}{2}\right \rfloor}}{2^m}\\
		&= \frac{R}{d+1},
	\end{align*}
	where the last equality holds for reasons similar to those in the proof of Theorem \ref{thm:rmlinub}. 
\end{proof}

Now, we shift our attention to permuted RM codes $\{\mathcal{C}_m^\pi(R)\}_{m\geq 1}$, defined by the sequence of permutations $(\pi_m)_{m\geq 1}$, with $\pi_m: [0:2^m-1]\to [0:2^m-1]$ (see the discussion preceding Theorem \ref{thm:genperminf} in Section \ref{sec:main}). Also recall the definition of $\overline{\mathcal{C}}_{d,\pi}^{(m)}$ be the largest \emph{linear} $(d,\infty)$-RLL subcode of $\mathcal{C}_m^\pi(R)$.

We shall now prove Theorem \ref{thm:genperminf}.
\begin{proof}[Proof of Theorem \ref{thm:genperminf}]
	We wish to prove that for ``most'' orderings, and for large $m$, it holds that the rate of $\overline{\mathcal{C}}_{d,\pi}^{(m)}$ is bounded above by $\frac{R}{d+1}+\delta_m$, where $\delta_m \xrightarrow{m\to \infty}0$.
	
	To this end, we first make the observation that the sequence of RM codes $\{\mathcal{C}_m(R) = \text{RM}(m,r_m)\}$ achieves a rate $R$ over the BEC, under block-MAP decoding too (see \cite{kud1} and \cite{kud3}). Hence, for large enough $m$, the (linear) RM code $\mathcal{C}_m(R)$ can correct erasures that are caused by a BEC$(1-R-\gamma_m)$, with $\gamma_m>0$, and $\gamma_m\xrightarrow{m\to \infty}0$. This then means that for large $m$, $\mathcal{C}_m(R)$ can correct $2^m(1-R-\gamma_m)-\delta\cdot \sqrt{2^m (1-R-\gamma_m)}$ erasures, with high probability (see Lemma 15 of \cite{abbesw}). Finally, from Corollary 18 of \cite{abbesw}, it then holds that for large enough $m$, any collection of $2^m R(1+\alpha_m)$ columns of $G_\text{Lex}(m,r_m)$, chosen uniformly at random, must have full row rank, $K_m$, with probabilty $1-\delta_m$, with $\alpha_m, \delta_m>0$ and $\alpha_m,\delta_m\xrightarrow{m\to \infty}0$.

	In other words, the discussion above implies that for large enough $m$, a collection of $K_m(1+\alpha_m)$ coordinates, chosen uniformly at random, contains an information set, with probabilty $1-\delta_m$. Viewing the above statement differently, it can be argued that for large enough $m$, for a $1-\delta_m$ fraction of the possible permutations $\pi_m: [0:2^m-1]\to [0:2^m-1]$, the first block of $K_m(1+\alpha_m)$ coordinates of the code $\mathcal{C}_{m}^\pi(R)$, contains an information set, ${\mathcal{J}}_{m,r_m}$. Now, within these ``good'' permutations, since $|{\mathcal{J}}_{m,r_m}| = K_m$, it follows that the number of runs, $\left\lvert\Gamma_{m,r_m}^\pi\right \rvert$, of consecutive coordinates that belong to ${\mathcal{J}}_{m,r_m}$, obeys $\left\lvert\Gamma_{m,r_m}^\pi\right \rvert \leq K_m\alpha_m$, with $\Gamma_{m,r_m}^\pi$ defined similar to equation \eqref{eq:gammam}. This is because, the number of runs, $\left\lvert\Gamma_{m,r_m}^\pi\right \rvert$, equals the number of coordinates $s$,  such that $\mathbf{B}(\pi_m(s))\notin {\mathcal{J}}_{m,r_m}$, but $\mathbf{B}(\pi_m(s-1))\in {\mathcal{J}}_{m,r_m}$, and the number of such $s$ is at most $K_m(1+\alpha_m)-K_m$, which equals $K_m\alpha_m$.
	
	Hence, the overall number, $t_m^\pi$, of disjoint $(d+1)$-tuples of consecutive coordinates in $\mathcal{J}_{m,r}$, satisfies (see the proof of Theorem \ref{thm:rmlinub})
	\[
	t_m^\pi\geq \frac{K_m}{d+1}-K_m\alpha_m,
	\]
	for a $1-\delta_m$ fraction of permutations $\pi_m$. Again, applying Proposition \ref{prop:linub}, it holds that for a $1-\delta_m$ fraction of permutations, with $\delta_m\xrightarrow{m\to \infty}0$, the rate of the largest $(d,\infty)$-RLL subcode obeys
	\begin{align*}
		\frac{\log_2\left \lvert \overline{\mathcal{C}}_{d,\text{G}}^{(m)}\right\rvert}{2^m}&\leq \frac{K_m-dt_m^\text{G}}{2^m}\\
		&\leq \frac{K_m-\frac{dK_m}{d+1}+dK_m\alpha_m}{2^m}\\
		&= \frac{R}{d+1}+o(1),
	\end{align*}
	thereby showing what we set out to prove.
\end{proof}
\section{Achievable Rates Using Cosets Of RM Codes}
\label{sec:cosets}
The results summarized in the previous sections provide lower and upper bounds on achievable rates by using subcodes of RM codes. In particular, Theorem \ref{thm:rm} (Theorem III.2 of \cite{arnk22arxiv}) shows that, using subcodes of RM codes, rates of up to $2^{-\left \lceil \log_2(d+1)\right \rceil} \cdot C$ are achievable over $(d,\infty)$-RLL input-constrained BMS channels. In this section, we provide another construction, which uses cosets of RM codes. The rates achieved by this construction, under bit-MAP decoding, are better than those in Theorem \ref{thm:rm}, for low noise regimes of the BMS channel. For example, for the case where $d=1$, the new coding scheme offers better rates for erasure probabilities $\epsilon\lessapprox 0.2837$, for the BEC, and for  crossover probabilities $p\in (0,0.0392)\cup(0.9608,1)$, for the BSC). In what follows, we set $N_m:=2^m$.

Fix a rate $R\in (0,C)$ and any sequence $\{{\mathcal{C}}_m(R) = \text{RM}(m,r_m)\}_{m\geq 1}$ that achieves a rate $R$ over the unconstrained BMS channel, under bit-MAP decoding. We interchangeably index the coordinates of any codeword in ${\mathcal{C}}_m(R)$ by $m$-tuples in the lexicographic order, and by integers in $[0:2^m-1]$. Recall, from Lemma \ref{lem:infoset}, that the set $\mathcal{I}_{m,r_m}:= \{\mathbf{b} = (b_1,\ldots,b_m)\in \mathbb{F}_2^m: \text{wt}(\mathbf{b})\leq r_m\}$ is an information set of ${\mathcal{C}}_m(R)$. For the remainder of this section, we let $m$ be a large positive integer.

We set $K_m=$ dim$({\mathcal{C}_m(R)}) = {m\choose \leq r_m}$. For large $m$, it holds that 
\begin{align}
	\label{eq:km}
	K_m &\in [(1-\alpha_m)N_mR,(1+\alpha_m)N_mR],\ \text{and}\notag\\
	N_m-K_m&\in [(1-\beta_m)N_m(1-R),(1+\beta_m)N_m(1-R)],
\end{align}
for $\alpha_m, \beta_m>0$, with $\alpha_m, \beta_m\xrightarrow{m\to \infty} 0$. 

For the purposes of our coding scheme, we shall work with specific permutations of the codes $\{{\mathcal{C}}_m(R)\}_{m\geq 1}$. Consider any permutation $\pi_m: [0:N_m-1]\to [0:N_m-1]$ with the property that $\pi_m([0:K_m-1]) = \mathcal{I}_{m,v_m}$, where, for a permutation $\sigma$, and a set $\mathcal{A}\subseteq [0:N_m-1]$, we define the notation $\sigma(\mathcal{A}):=\{\sigma(i):i\in \mathcal{A}\}$. As in Section \ref{sec:main}, we define the permuted code ${\mathcal{C}}^{\pi}_m(R)$ as
\begin{align*}
{\mathcal{C}}^{\pi}_m(R) = \big\{(c_{\pi_m(0)},c_{\pi_m(1)}&,\ldots,c_{\pi_m(N_m-1)}):\\ &(c_0,c_1,\ldots,c_{N_m-1})\in {\mathcal{C}}_m(R)\big\}.
\end{align*}

Thus, ${\mathcal{C}}^{\pi}_m(R)$ is the code obtained by permuting the coordinates of codewords in ${\mathcal{C}}_m(R)$, such that the coordinates in the information set $\mathcal{I}_{m,r_m}$ occur in the first block of $K_m$ positions. Note that the permuted code ${\mathcal{C}}^{\pi}_m(R)$ is systematic, in that all possible $K_m$-tuples of binary symbols can occur in its first $K_m$ coordinates, and in particular, all $K_m$-tuples that respect that $(d,\infty)$-RLL constraint, occur in these coordinates. We let $G^\pi_m$ be a \emph{systematic} generator matrix for ${\mathcal{C}}^{\pi}_m(R)$. For the lemma that follows, we shall use the notation 

\begin{align*}\tilde{\mathcal{C}}^\pi_m:=\big\{(\tilde{c}_{\pi_m(0)},&\tilde{c}_{\pi_m(1)},\ldots,\tilde{c}_{\pi_m(N_m-1)}:\\ &(\tilde{c}_0,\tilde{c}_1,\ldots,\tilde{c}_{N_m-1})\in \tilde{\mathcal{C}}(m,r_m)\big\},\end{align*}
where $\tilde{\mathcal{C}}(m,r_m) = \text{span}\left(\mathcal{B}_{m,r_m}\right)$ (see equation \eqref{eq:Bmr} in Section \ref{sec:rmublin}). 
\begin{lemma}
	\label{lem:coset}
	For every codeword $\mathbf{w}\in {\mathcal{C}}^{\pi}_m(R)$, there exists a vector $\mathbf{v}\in \tilde{\mathcal{C}}^\pi_m$, such that $\mathbf{w}+\mathbf{v}$ (over $\mathbb{F}_2$) equals the concatenation $w_1^{K_m}\mathbf{0}$.
\end{lemma}
\begin{proof}
	The proof is a simple consequence of Lemma \ref{lem:quotient}. Indeed, since the last $N_m-K_m$ coordinates in the permuted code ${\mathcal{C}}^{\pi}_m(R)$ are exactly those coordinates $\mathbf{b}\in \{0,1\}^m$ such that $\mathbf{b}\notin  \mathcal{I}_{m,r_m}$, i.e., with wt$(\mathbf{b})\geq r_m+1$, we have that any standard basis vector with a $1$ in these coordinates belongs to $\tilde{\mathcal{C}}^\pi_m$. By taking suitable linear combinations of these standard basis vectors, it is possible to obtain a word $\mathbf{v}\in \tilde{\mathcal{C}}^\pi_m$ such that $w_{K_m+1}^{N_m} = v_{K_m+1}^{N_m}$, with $v_{1}^{K_m} = \mathbf{0}$. Hence, it holds that $\mathbf{w}+\mathbf{v} = w_1^{K_m}\mathbf{0}$, over $\mathbb{F}_2$.
\end{proof}
\begin{remark}
	Note that words $\mathbf{v} \in \tilde{\mathcal{C}}^\pi_m$, which are of the form $\mathbf{v} = 0^{K_m} v_{K_m+1}^{N_m}$, for some $v_{K_m+1},\ldots,v_{N_m}\in \{0,1\}$,
are in one-to-one correspondence with the cosets of ${\mathcal{C}}^{\pi}_m(R)$. In other words, each word in $\tilde{\mathcal{C}}^\pi_m$ uniquely identifies a coset of ${\mathcal{C}}^{\pi}_m(R)$. In what follows, we consider $ \tilde{\mathcal{C}}^\pi_m$ to be the collection of coset leaders for the code ${\mathcal{C}}^{\pi}_m(R)$.
\end{remark}

We now describe a simple encoding strategy to transmit $(d,\infty)$-RLL input-constrained words over the BMS channel:
\begin{enumerate}
	\item Pick a $(d,\infty)$-RLL constrained $K_m$-tuple,  $w_1^{K_m}$. Encode $w_1^{K_m}$ into a codeword $\mathbf{c} \in \mathcal{C}_m^\pi(R)$, using the systematic generator matrix $G_m^\pi$, with $ \mathbf{c} = w^{K_m}_1 G_m^\pi$. Note that $c_1^{K_m} = w_1^{K_m}$.
	\item Choose a coset leader $\mathbf{v}\in \tilde{\mathcal{C}}^\pi_m$ such that the word, $\mathbf{c}+\mathbf{v} = w_1^{K_m}\mathbf{0}$, is also $(d,\infty)$-RLL constrained.
	\item Transmit the first $K_m$ bits, $w_1^{K_m}$, of $\mathbf{c}+\mathbf{v}$.
	\item Transmit the identity of the coset leader.
\end{enumerate}
Choosing an RLL constrained word in Step 1 above can be accomplished using well-known constrained encoders (see, for example, \cite{adler} and Chapters 4 and 5 of \cite{Roth}), of rates arbitrarily close to the noiseless capacity, $C_0^{(d)}$, of the $(d,\infty)$-RLL constraint. Further, Lemma \ref{lem:coset} shows that Step 2 can also be achieved. Step 4 will be explained further below. At the decoder end, the coset leader $\mathbf{v}$ is recovered first, and this information is used to decode the original codeword, $\mathbf{c} \in  {\mathcal{C}}^{\pi}_{m}(R)$.

We now elaborate on Step 4, in more detail. Our objective is to use extra channel uses that encode the last $N_m-K_m$ bits of $\mathbf{v}$, which uniquely identify the coset leader, into a $(d,\infty)$-RLL input-constrained word, and transmit this input-constrained word to the decoder. Observe, from Lemma \ref{lem:coset}, that $v_{K_m+1}^{N_m} = c_{K_m+1}^{N_m}$. 

Now, by Theorem \ref{thm:rm}, we can identify $(d,\infty)$-RLL subcodes of RM codes of rate $R$, which achieve rates of up to $2^{-\left \lceil \log_2(d+1)\right \rceil}\cdot R$. We shall use these subcodes to encode the last $N_m-K_m$ bits of $\mathbf{v}$. We mention that since $m$ is large, the rate of the subcode ${\mathcal{C}}_m^{(d,\infty)}(R)$ (see Theorem \ref{thm:rm}), which we write as $R_m^{(d,\infty)}$, obeys
\begin{align}
	\label{eq:rmdinf}
	R_m^{(d,\infty)}\in \left[2^{-\left \lceil \log_2(d+1)\right \rceil}\cdot R (1-\gamma_m),2^{-\left \lceil \log_2(d+1)\right \rceil}\cdot R (1+\gamma_m)\right],
\end{align}
for $\gamma_m>0$, with $\gamma_m\xrightarrow{m\to \infty} 0$.

Since our objective is to encode the $N_m-K_m$ bits identifying the coset leader using subcodes of RM codes, we require that the blocklength after encoding is a power of $2$. To facilitate this, we first divide the $N_m-K_m$ bits to be encoded into smaller parts, each of which will be separately encoded into a $(d,\infty)$-RLL constrained codeword of an RM code.

In particular, having chosen a large $m$, we identify a large, fixed, positive integer $\tau$, and a positive integer $L$, such that 
\begin{align}
	\label{eq:L}
	\left[\frac{(1-R)(1-\beta_m)}{R(1+\gamma_m)}, \frac{(1-R)(1+\beta_m)}{R(1-\gamma_m)}\right] \subseteq \left[\frac{L-1}{2^\tau}, \frac{L}{2^\tau}\right].
\end{align}

We chop up the last $N_m-K_m$ bits of $\mathbf{v}$ into $L$ equal parts, with each part having $\frac{N_m-K_m}{L}$ bits (see the remark below). We shall use $(d,\infty)$-RLL subcodes of RM codes to now encode each of these $L$ parts. We then have that the number of channel uses, $N_{\text{exact}}$, needed to transmit each part using a $(d,\infty)$-RLL RM subcode of rate $R_m^{(d,\infty)}$, is $\frac{N_m-K_m}{L\cdot R_m^{(d,\infty)}}$, which from equations \eqref{eq:km}--\eqref{eq:L}, satisfies
\begin{align}
	N_{\text{exact}} &\leq 2^{m-\tau+\left \lceil \log_2(d+1)\right \rceil}\notag\\
	&=: N_{\text{part}}, \label{eq:part}
\end{align}
where we have used the fact that $N_m=2^m$. Since we need the blocklength to be a power of $2$, we use $N_{\text{part}}$ channel uses to transmit each of the $L$ parts into which the $N_m-K_m$ bits have been divided. 
The total number of channel uses needed to convey the identity of the coset leader is thus $N_{\text{part}}\cdot L$. We set $n:= m-\tau+\left \lceil \log_2(d+1)\right \rceil$, with $N_{\text{part}} = 2^{n}$. Thus, step 4 of the encoding strategy can be replaced by the following two steps:

\begin{itemize}
	\item [4a)] Divide $c_{K_m+1}^{N_m}$ into $L$ equal parts, ${\mathbf{c}}_1,\ldots, {\mathbf{c}}_L$.
	\item [4b)] Encode each part $\mathbf{c}_i$, for $i\in [L]$, into a codeword of the code ${\mathcal{C}}^{(d,\infty)}_{n}(R)$ (see equation \eqref{eq:rmlb1}), of blocklength $2^n = N_{\text{part}}$.
\end{itemize} 
\begin{remark}
	For ease of reading, we assume that $m$ is such that $L$ divides $N_m-K_m$. However, the general case can be handled by appending at most $L-1$ $0$s at the end of the $N_m-K_m$ bits, so that the overall length is divisible by $L$, thereby giving rise to the same lower bound in Lemma \ref{lem:ratecoset}
\end{remark}

The construction of our code $\mathcal{C}^{\text{cos}}_m$ is given in Algorithm \ref{alg:coset}, with the assumption that $L$ divides $N_m-K_m$.
We let a generator matrix of the linear code ${\mathcal{C}}^{(d,\infty)}_{m}(R)$ in Theorem \ref{thm:rm} be denoted by $G_{m}^{(d)}$.

\begin{algorithm}[h]
	\caption{Construction of $(d,\infty)$-RLL constrained code $\mathcal{C}^{\text{cos}}_m$}
	\label{alg:coset}
	\begin{algorithmic}[1]	
		\Procedure{Coding-Scheme}{$G^{\pi}_m$, $G_n^{(d)}$}       
		\State Pick a $(d,\infty)$-RLL constrained $K_m$-tuple $w_1^{K_m}$.
		\State Obtain $\mathbf{c} \in \mathcal{C}_m^\pi(R)$ as $\mathbf{c} = w^{K_m}_1 G_m^\pi$, with $c_1^{K_m} = w_1^{K_m}$.
		\State Set $\mathbf{x}_1:= w_1^{K_m}$.
		\State Divide $c_{K_m+1}^{N_m}$ into $L$ equal parts, ${\mathbf{c}}_1,\ldots, {\mathbf{c}}_L$.
		\For{$i=1:L$}
		\State Set $\mathbf{x}_{2,i} = {\mathbf{c}}_i G_n^{(d)}$.
		\EndFor
		\State Set $\mathbf{x}_2=\mathbf{x}_{2,1}\ldots \mathbf{x}_{2,L}$.
		\State Transmit $\mathbf{x} = \mathbf{x}_1\mathbf{x}_2$.
		\EndProcedure
		
	\end{algorithmic}
\end{algorithm} 

We note from the construction of ${\mathcal{C}}^{(d,\infty)}_{m}(R)$ in \eqref{eq:rmlb1} that the first $d$ symbols in $\mathbf{x}_{2,i}$ are $0$s, for all $i\in [L]$. Hence, the $(d,\infty)$-RLL input constraint is satisfied at the boundaries of the concatenations in steps 8 and 9, too.

The rate of the coding scheme in Algorithm \ref{alg:coset} is summarized in the lemma below.
\begin{lemma}
	\label{lem:ratecoset}
	The rate of the coding scheme in Algorithm \ref{alg:coset} satisfies
	\[
	\liminf_{m\to \infty} \text{rate}(\mathcal{C}_m^{\text{cos}})\geq \frac{C_0^{(d)}\cdot R^2\cdot 2^{-\left \lceil \log_2(d+1)\right \rceil}}{R^2\cdot 2^{-\left \lceil \log_2(d+1)\right \rceil} + 1-R+2^{-\tau}},
	\]
where $C_0^{(d)}$ is the noiseless capacity of the $(d,\infty)$-RLL input constraint, and $\tau$ is an arbitrarily large, fixed, positive integer.
\end{lemma}
\begin{proof}
	Recall that the noiseless capacity, $C_0^{(d)}$, of the $(d,\infty)$-RLL constraint, is given by (see, for example, \cite{Roth})
	\begin{align}
		C_0^{(d)} &= \lim_{n\to \infty}\frac{\log_2|S_{(d,\infty)}^{(n)}|}{n} \notag\\
		&= \inf_n \frac{\log_2|S_{(d,\infty)}^{(n)}|}{n}, \label{eq:inf}
	\end{align}
	where the last equality follows from the subadditivity of the sequence $\left(\log_2|S_{(d,\infty)}^{(n)}|\right)_{n\geq 1}$. 
	
	By picking $m$ large enough (and hence $K_m$ large enough), we note that for step 2 of Algorithm 1, there exist constrained coding schemes (see \cite{adler} and Chapters 4 and 5 of \cite{Roth}) of rate $C_0^{(d)}- \epsilon_m$, for $\epsilon_m>0$, with $\epsilon_m\xrightarrow{m\to \infty} 0$. Hence, we see that for large $m$, the number of possible $K_m$-tuples, $w_1^{K_m}$, that can be picked, equals $2^{K_m(C_0^{(d)}-\epsilon_m)}$. Since the codeword $\mathbf{c}$ and the words $\mathbf{x}_1$ and $\mathbf{x}_2$ are determined by $\mathbf{w}$, it holds that for large $m$, the rate of the code $\mathcal{C}_m^{\text{cos}}$ obeys
	\begin{align*}
	\label{eq:ratecalc}
	\text{rate}(\mathcal{C}_m^{\text{cos}}) \geq \frac{\log_2\left(2^{K_m(C_0^{(d)}-\epsilon_m)}\right) }{K_m+N_{\text{part}}\cdot L},
\end{align*}
where the denominator, $K_m+N_{\text{part}}\cdot L$, is the total number of channel uses. The following statements then hold true:
\begin{align*}
	\text{rate}(\mathcal{C}_m^{\text{cos}}) &\geq \frac{\log_2\left(2^{K_m(C_0^{(d)}-\epsilon_m)}\right)}{K_m+N_{\text{part}}\cdot L}\\
	&\stackrel{(a)}{=} \frac{\frac{\left(C_0^{(d)}-\epsilon_m\right)\cdot{K_m}}{N_m}}{\frac{K_m}{N_m}+L\cdot 2^{-\tau+\left \lceil \log_2(d+1)\right \rceil}}\\
	&\stackrel{(b)}{\geq} \frac{\frac{\left(C_0^{(d)}-\epsilon_m\right)\cdot{K_m}}{N_m}}{\frac{K_m}{N_m}+2^{\left \lceil \log_2(d+1)\right \rceil}\cdot \left(\frac{(1-R)(1-\beta_m)}{R(1+\gamma_m)}+2^{-\tau}\right)},
\end{align*}
where (a) follows from equation \eqref{eq:part} and (b) holds due to equation \eqref{eq:L}, with $L\cdot 2^{-\tau}\leq \left(\frac{(1-R)(1-\beta_m)}{R(1+\gamma_m)}+2^{-\tau}\right)$. Hence, by taking $\liminf_{m\to \infty}$ on both sides of the inequality (b) above, we get
\begin{align*}
	\liminf_{m\to \infty} \text{rate}(\mathcal{C}_m^{\text{cos}}) &\geq \frac{C_0^{(d)}\cdot R}{R+2^{\left \lceil \log_2(d+1)\right \rceil}\cdot \left(\frac{1-R}{R}\right)+2^{\left \lceil \log_2(d+1)\right \rceil-\tau}}\\
	&= \frac{C_0^{(d)}\cdot R^2\cdot 2^{-\left \lceil \log_2(d+1)\right \rceil}}{R^2\cdot 2^{-\left \lceil \log_2(d+1)\right \rceil} + 1-R+2^{-\tau}},
\end{align*}
where the inequality holds since $\frac{K_m}{N_m}\xrightarrow{m\to \infty}R$ and $\epsilon_m,\beta_m, \gamma_m\xrightarrow{m\to \infty}0$.

\end{proof}
The proof of Theorem \ref{thm:rmcosets} follows by noting that any rate $R\in (0,C)$ is achievable by RM codes over an unconstrained BMS channel, under bit-MAP decoding, and by substituting $C$ instead of $R$ in Lemma \ref{lem:ratecoset}.

\section{Conclusion}
\label{sec:conclusion}
In this paper, we derived upper bounds on the rates of linear $(d,\infty)$-RLL subcodes of Reed-Muller (RM) codes. Our work, therefore, provides upper bounds on achievable rates using linear subcodes of RM codes, over binary memoryless symmetric (BMS) channels with $(d,\infty)$-RLL constrained inputs. 
We showed that if $C$ is the capacity of an unconstrained BMS channel, then the rate of any linear $(d,\infty)$-RLL subcode of an RM code, is bounded above by $\frac{C}{d+1}$, in the limit as the blocklength of the code goes to infinity. A discussion about RM codes under coordinate orderings different from the lexicographic ordering was also taken up. In particular, we showed that for linear $(d,\infty)$-RLL subcodes of RM codes under a Gray ordering, the same upper bound holds, and that for large enough blocklength, for nearly all coordinate orderings, a rate upper bound of $\frac{C}{d+1}+\delta$ holds, where $\delta$ can be taken to be as small as required.
Further, we devised a constrained coding scheme based on cosets of RM codes that, for low noise regimes, outperforms any linear coding scheme, in terms of rate. For values of $C$ close to $1$, the rate of our coding scheme is also close to the coset-averaging bound of \cite{pvk}.




For future work, as regards the cosets-based coding scheme proposed in this paper, other sequential decoding algorithms (such as those in \cite{honda}), adapted to RM codes, can be explored to check if the need for extra channel uses, for exchanging coset information, can be eliminated altogether.
\section{Acknowledgements}
The authors would like to thank Prof. Henry Pfister for stimulating discussions.



\ifCLASSOPTIONcaptionsoff
  \newpage
\fi



%
\bibliographystyle{IEEEtran}
{\footnotesize
	\bibliography{references}}

%

%





\end{document}